%% file: main.tex
\documentclass{article}

\usepackage[utf8]{inputenc} 
\usepackage[T1]{fontenc}    
\usepackage{url}            
\usepackage{booktabs}       
\usepackage{amsfonts}       
\usepackage{nicefrac}       
\usepackage{microtype}      
\usepackage{xcolor}         

\input{header}

\title{Accelerated Relax-and-Round for Concave Coverage Problems}

\author[ \hspace{-1.0ex}]{Matthew Fahrbach\footnote{Equal contribution. Emails: \texttt{\{fahrbach,mehraneh,zadim\}@google.com}}}
\author[ \hspace{-1.0ex}]{Mehraneh Liaee\protect\footnotemark[1]}
\author[ \hspace{-1.0ex}]{Morteza Zadimoghaddam}

\affil[ \hspace{-1ex}]{Google Research}

\date{}

\usepackage{tocloft}
\begin{document}

\maketitle

\input{abstract}

\newpage
\tableofcontents

\newpage
\input{introduction}
\input{preliminaries}
\input{smooth_proxy_objective}
\input{algorithm}
\input{approximation_ratios}
\input{experiments}

{
\bibliographystyle{plainnat}
\bibliography{references}
}

\newpage
\appendix
\input{app_introduction}
\input{app_smooth_proxy_objective}
\input{first_order_methods}
\input{app_convex_combination}

\input{app_rounding_ratio}
\input{app_greedy}
\input{app_experiments}

\end{document}

%% file: header.tex
\usepackage[utf8]{inputenc}
\usepackage[T1]{fontenc}
\usepackage{inconsolata}
\usepackage{microtype}
\usepackage{algorithm}
\usepackage[noend]{algpseudocode}
\usepackage{amsmath, amssymb, amsthm}
\usepackage{bm}
\usepackage{upgreek}

\usepackage{graphicx}
\usepackage{thmtools}
\usepackage{thm-restate}
\usepackage[backref=page]{hyperref}  
\usepackage[capitalize, nameinlink]{cleveref}
\usepackage{comment}
\usepackage{enumerate}
\usepackage{mathtools}
\usepackage{subcaption}
\usepackage[dvipsnames]{xcolor}
\usepackage{nag}
\usepackage{stmaryrd}
\usepackage{tcolorbox}
\usepackage{xfrac}
\usepackage{xspace}
\usepackage{booktabs}
\usepackage{multirow}
\usepackage{lipsum}

\usepackage[numbers]{natbib}
\usepackage{authblk}

\usepackage{geometry}
\usepackage{tikz}
\usetikzlibrary{shapes, arrows.meta, positioning, decorations.pathreplacing}

\graphicspath{{figures/}{.}}

\newcommand{\declarecolor}[2]{\definecolor{#1}{RGB}{#2}\expandafter\newcommand\csname #1\endcsname[1]{\textcolor{#1}{##1}}}
\declarecolor{White}{255, 255, 255}
\declarecolor{Black}{0, 0, 0}
\declarecolor{LightGray}{216, 216, 216}
\declarecolor{Gray}{127, 127, 127}
\declarecolor{Orange}{237, 125, 49}
\declarecolor{LightOrange}{251,229, 214}
\declarecolor{Yellow}{255, 192, 0}
\declarecolor{LightYellow}{255, 242, 200}
\declarecolor{Blue}{91, 155, 213}
\declarecolor{LightBlue}{222, 235, 247}
\declarecolor{Green}{112, 173, 71}
\declarecolor{LightGreen}{226, 240, 217}
\declarecolor{Navy}{68, 114, 196}
\declarecolor{LightNavy}{218, 227, 243}

\hypersetup{
	colorlinks=true,
	pdfpagemode=UseNone,
	citecolor=Navy,
	linkcolor=Navy,
	urlcolor=Navy,
}

\crefformat{equation}{Eq.~(#2#1#3)}

\newcommand{\mytilde}{\raise.17ex\hbox{$\scriptstyle\mathtt{\sim}$}}

\newcommand{\ConcaveCoverage}{$\varphi$-\textsc{MaxCoverage}\xspace}
\newcommand{\xgreedy}{\mat{x}_{\texttt{greedy}}}
\newcommand{\degreeR}{d_R}
\newcommand{\MainAlgorithm}{\texttt{AcceleratedRelaxAndRound}\xspace}
\newcommand{\SwapRound}{\texttt{SwapRound}\xspace}
\newcommand{\MergeBases}{\texttt{MergeBases}\xspace}
\newcommand{\SmoothC}{\widetilde{C}_\mu}

\newcommand{\mat}[1]{\mathbf{#1}}

\newcommand{\OPT}{\textnormal{OPT}}
\newcommand{\ALG}{\textnormal{ALG}}
\newcommand{\Z}{\mathbb{Z}}
\newcommand{\R}{\mathbb{R}}
\newcommand{\E}{\mathbb{E}}
\DeclareMathOperator*{\argmax}{arg\,max}
\DeclareMathOperator*{\argmin}{arg\,min}
\newcommand{\Var}{\mathrm{Var}}

\newcommand{\Pois}{\textnormal{Pois}}

\newcommand{\rounded}{\textnormal{int}}
\newcommand{\facebook}{\texttt{facebook}\xspace}
\newcommand{\dblp}{\texttt{dblp}\xspace}

\newcommand{\defeq}{:=}

\DeclarePairedDelimiter{\abs}{\lvert}{\rvert}
\DeclarePairedDelimiter{\set}{\{}{\}}
\DeclarePairedDelimiter{\parens}{(}{)}
\DeclarePairedDelimiter{\bracks}{[}{]}
\DeclarePairedDelimiter{\floor}{\lfloor}{\rfloor}
\DeclarePairedDelimiter{\ceil}{\lceil}{\rceil}
\DeclarePairedDelimiter{\norm}{\lVert}{\rVert}

\theoremstyle{plain}
\newtheorem{theorem}{Theorem}[section]
\newtheorem{lemma}[theorem]{Lemma}
\newtheorem{corollary}[theorem]{Corollary}

\theoremstyle{definition}
\newtheorem{definition}[theorem]{Definition}

\newtheorem{remark}[theorem]{Remark}

%% file: abstract.tex
\begin{abstract}
We present an accelerated relax-and-round algorithm for concave coverage problems, which generalize the classic maximum coverage problem.
Building on the relax-and-round framework of Barman et al.~[\hyperlink{cite.barman2021tight}{STACS 2021}],
we propose two significant improvements.
First, we replace the linear programming (LP) relaxation step with a projected accelerated gradient method
applied to a smooth surrogate objective to achieve a $\tilde{O}(mn \varepsilon^{-1})$ running time.
Second, we use a specialized rounding scheme for the hypersimplex that combines the
Carathéodory decomposition algorithm in~Karalias et al.~[\hyperlink{cite.karalias2025geometric}{NeurIPS 2025}]
with randomized swap rounding of~Chekuri et al.~[\hyperlink{cite.chekuri2010dependent}{FOCS 2010}].
We prove tight approximation ratios for new reward functions,
including a $0.827$-approximation for the logarithmic reward $\varphi(x) = \log(1 + x)$.
Finally, we conduct maximum multi-coverage experiments on synthetic and real-world graphs,
demonstrating that our algorithm outperforms approaches that use state-of-the-art LP solvers.
\end{abstract}

%% file: introduction.tex
\section{Introduction}
\label{sec:introduction}

While the massive scale of modern datasets continues to push the machine learning (ML) frontier,
the corresponding costs of data curation, storage, and model training have increased proportionally.
Consequently, efficient data pruning has become a critical component of ML workflows.
Data subset selection is typically formulated as a combinatorial optimization problem where the objective is to find a core subset that fully represents the dataset while minimizing redundancy.
This approach is widely used across domains, including
active learning~\citep{wei2015submodularity,sener2018active,coleman2020selection,yehuda2022active,fahrbach2025gist},
machine translation~\citep{kirchhoff2014submodularity,biccici2014optimizing},
bioinformatics~\citep{hobohm1992selection,libbrecht2018choosing},
and data mixture optimization~\citep{renduchintala2024smart, sachdeva2024train, thudi2025mixmin}.
Motivated by these applications, we develop scalable algorithms for the \emph{concave coverage problem},
a generalization of the NP-hard maximum coverage problem that provides strong theoretical guarantees.

\begin{definition}[Concave coverage]
Given a concave, nondecreasing function $\varphi : \Z_{\ge 0} \rightarrow \R_{\ge 0}$,
a bipartite graph $G = (L, R, E)$ with $|L| = n$ and $|E| = m$,
nonnegative weights $w_{j}$ for each $j \in R$,
and a cardinality constraint $k \in \Z_{\ge 0}$,
the objective of \ConcaveCoverage is to maximize
\begin{equation}
\label{eqn:objective_function}
  C_{\varphi}(S) := \sum_{j \in R} w_j \varphi(|S \cap N(j)|),
\end{equation}
over subsets $S \subseteq L$ such that $|S| \le k$.
Without loss of generality, we assume $\varphi$ is normalized such that $\varphi(0) = 0$ and $\varphi(1) = 1$,
that $\sum_{j \in R} w_{j} = 1$, and that $\deg(j) \ge 1$ for all $j \in R$.
\end{definition}

\begin{definition}
The \emph{Poisson concavity ratio} of $\varphi$ is
\begin{equation}
\label{eqn:poisson_concavity_ratio}
  \alpha_{\varphi}
  := 
  \inf_{x \in \Z_{\ge 1}}
  \frac{\E[\varphi(\Pois(x))]}{\varphi(\E[\Pois(x)])}
  =
  \inf_{x \in \Z_{\ge 1}}
  \frac{\E[\varphi(\Pois(x))]}{\varphi(x)},
\end{equation}
where $\Pois(x)$ is a Poisson random variable with rate $x$.
We write $\alpha_{\varphi}(x) := \E[\varphi(\Pois(x))]/\varphi(x)$.
\end{definition}

To give an example, we can frame data pruning as a \emph{$c$-multi-coverage problem} on a bipartite graph~\citep{kirchhoff2014submodularity}.
Let $L$ be a set of $n$ text documents and $R$ be a set of semantic topics (e.g., concepts in a natural language processing task).
Define the edge $\{i,j\} \in E$ if document $i \in L$ represents topic $j \in R$
(e.g., an ESPN article connects to the ``sports'' topic).
To ensure a model learns each concept sufficiently, our goal is to select a subset of data $S \subseteq L$ with budget $|S| \le k$
such that every topic is covered at least $c$ times.
This objective is captured by the concave function $\varphi(x) = \min(x, c)$.
Since concave coverage is an instance of monotone submodular maximization subject to a cardinality constraint,
the greedy algorithm guarantees a $(1-\sfrac{1}{e})$-approximate solution~\citep{nemhauser1978analysis}.
\citet{barman2021tight, barman2022tight}, however, introduced a relax-and-round algorithm for concave coverage
and proved it has a tight approximation ratio equal to the Poisson concavity ratio $\alpha_{\varphi}$,
which is purely a function of $\varphi$ and often better than $1 - \sfrac{1}{e}$.
For example, if we require deep semantic coverage with $c = 1000$, then $\alpha_{\varphi} =  1 - c^c / (e^c c!) \ge 0.987$.
Conversely, setting $c = 1$ recovers the classic maximum coverage problem and its $1-\sfrac{1}{e}$ guarantee.

\begin{theorem}[{\citet{barman2021tight}}]
\label{thm:barman_approximation_ratio}
For any normalized nondecreasing concave function $\varphi$,
there is a polynomial-time $\alpha_{\varphi}$-approximation
algorithm for the \ConcaveCoverage problem.
Furthermore, for $\varphi(n) = o(n)$, it is NP-hard to approximate
the \ConcaveCoverage problem to within a factor of
$\alpha_\varphi + \varepsilon$, for any constant $\varepsilon > 0$.
\end{theorem}

\noindent
We note that for any reward function $\varphi$, we have $\alpha_{\varphi} \ge 1 - \sfrac{1}{e}$.
We give a short proof in~\Cref{app:poisson_concavity_ratio_lower_bound}.
\vspace{-0.5cm}

\begin{table*}
    \centering
    \caption{Running times and approximation ratios of concave coverage algorithms,
    where $\omega \approx 2.371$~is the matrix multiplication constant
    and $\Delta \hspace{-0.045cm}=\hspace{-0.045cm} \max_{j \in R} \deg(j)$ is the maximum degree of the right nodes.}
    \label{tab:theory_comparison}
    \begin{tabular}{l c c c c c}
        \toprule
        \textbf{Algorithm} & \textbf{Running time} & \textbf{Approximation}\\ 
        \midrule
        Greedy & $O(mk)$ & $1 - \sfrac{1}{e}$ \\
        \citet{barman2021tight} & $\tilde{O}((n + |R|)^{\omega} \log(1/\varepsilon) + n m \Delta)$ & $\alpha_{\varphi} - \varepsilon$ \\
        This paper & $O\parens*{(m + n \log n) \frac{n \sqrt{\log n}}{\varepsilon}}$ & $\alpha_{\varphi} - \varepsilon$ \\
        \bottomrule
    \end{tabular}
    \vspace{-0.25cm}
\end{table*}

\subsection{Our contributions}

We summarize the main contributions of our work below:

\begin{itemize}
    \item \emph{Faster LP relaxation.}
    We construct a smooth surrogate coverage objective $\SmoothC(\mat{x})$ in \Cref{sec:smooth_proxy_objective}
    and run accelerated projected gradient ascent starting from $\mat{x}^{(0)} = \mat{x}_{\text{greedy}}$ (\Cref{alg:accelerated_relax_and_round})
    to optimize this concave proxy for $\mat{x} \in \Delta_{n,k}$.
    We bound the number of steps required for a
    $(1-\varepsilon)$-approximation to the original objective $C(\mat{x})$ in \Cref{thm:relative_error_algorithm},
    achieving a $\tilde{O}(mn/\varepsilon)$ running time,
    which is strictly faster than \citet{barman2021tight} (see \Cref{app:runtime_comparison}).

    \item \emph{Rounding.}
    After computing an approximate LP solution $\widetilde{\mat{x}}^*$,
    we use a rounding scheme tailored for points in the hypersimplex $\Delta_{n,k}$.
    Pipage rounding in \citet{barman2021tight} takes $O(nm \Delta)$ time,
    but by combining the Carathéodory decomposition in \citet{karalias2025geometric}
    and swap rounding~\citep{chekuri2010dependent},
    we round $\widetilde{\mat{x}}^*$ to a high-value integral feasible solution in $O(nk \log n)$ time
    and remove the dependence on the number of edges $m$ and maximum degree $\Delta$.

    \item \emph{Approximation ratios.}
    We study new reward functions and prove tight approximation ratios in \Cref{sec:approximation_ratios}.
    For example, if $\varphi(x) = \log(1 + x)$ then $\alpha_{\varphi} \approx 0.827$.
    We also study piecewise linear functions
    and the isoelastic utility $\varphi(x) = x^{1 - \gamma}$, which we show has a closed-form integral expression for $\alpha_{\varphi}$
    by building on the infinite series in \citep{barman2021tight}.

    \item \emph{Experiments.}
    We construct $c$-multi-coverage instances where the greedy algorithm gets stuck at $(1-\sfrac{1}{e})\OPT$
    but relax-and-round converges to the optimal solution, i.e., $\alpha_{\varphi} \rightarrow 1$.
    Then we compare \Cref{alg:accelerated_relax_and_round} to state-of-the-art LP solvers with pipage rounding on real-world SNAP graphs,
    demonstrating better end-to-end running times and rounded objective values.

\end{itemize}

\subsection{Related work}

\paragraph{Submodular maximization.}
Submodular maximization subject to a cardinality constraint $k$ is an NP-hard problem.
For monotone submodular functions, \citet{nemhauser1978analysis} showed that the greedy
algorithm achieves a $(1-\sfrac{1}{e})$-approximation, which is asymptotically optimal unless P = NP~\citep{feige1998threshold}.
Over the last decade, there has been significant research on distributed algorithms for submodular maximization
in the MapReduce model~\citep{mirzasoleiman2013distributed,mirrokni2015randomized,barbosa2015power,liu2019submodular} and
low-adaptivity models~\citep{balkanski2019exponential, ene2019submodular, fahrbach2019submodular,chen2024scalable}.

\paragraph{Concave coverage.}
\citet{barman2021tight,barman2022tight} study the concave coverage problem
and give a relax-and-round algorithm with tight approximation guarantees in terms of the
Poisson concavity ratio $\alpha_{\varphi}$.
A related line of work starts with \citet{stobbe2010efficient},
who give an algorithm for \emph{minimizing} a class of functions they call \emph{decomposable submodular functions},
i.e., those that can be represented as a sum of concave functions applied to modular functions (SCMs).
\citet{dolhansky2016deep} later studied \emph{deep submodular functions} (DSFs)
and showed that deep neural networks with nonnegative weights correspond to a DSF
and that DSFs strictly generalize SCMs.

\citet{bai2018submodular} is the first work to \emph{maximize} DSFs, a class of submodular functions that includes concave coverage.
Their algorithm uses projected gradient ascent and pipage rounding,
and requires $O(n^2 / \varepsilon^2)$ calls to a \emph{submodular value oracle} before rounding.
In an independent work, \citet{lv2026efficient} recently used accelerated projected gradient ascent and pipage rounding for SCM maximization.
However, neither guarantees worst-case constant-factor approximation ratios.
In fact, \citet{lv2026efficient} state that ``the approximation bounds above are not fixed constants,''
and specifically identify
$\varphi(x) = \log(1 + x)$ as one of ``several commonly used [concave] functions.'' 
Our work addresses these open problems with tight analysis using the framework of \citet{barman2021tight}.

\paragraph{ML applications.}
\citet{kirchhoff2014submodularity} use \emph{feature-based submodular functions} for a machine translation
task (an instance of concave coverage)
and find that ``the square root $\varphi(x) = \sqrt{x}$ performs slightly better than the log function'' in their experiments.
\citet{wei2014unsupervised} use the same formulation for a large-scale speech data subset selection task.
\citet{karimi2017stochastic} study weighted coverage functions and
give a concave relax-and-round algorithm for influence maximization.
\citet{deng2025procurement} design reverse data auctions via submodular maximization
and ``quasi-linear utilities,'' for which we derive tight and novel approximation ratios in \Cref{thm:piecewise_linear_ratio}.
For more examples, we direct the reader to~\citep{bilmes2022submodularity} for an excellent survey on applications of submodularity in machine learning.

%% file: preliminaries.tex
\section{Preliminaries}
\label{sec:preliminaries}

\paragraph{Notation.}
Let $C(S) \defeq C_{\varphi}(S)$ and $[n] \defeq \{1, 2, \dots, n\}$.
We use boldface for vectors $\mat{x} \in \R^n$
and matrices $\mat{A} \in \R^{m \times n}$,
with iterates of our projected gradient algorithm indexed by superscripts $\mat{x}^{(t)}$.
Let $\mat{1}_{S} \in \{0, 1\}^n$ be the indicator vector for $S \subseteq [n]$.
We use $\log(x)$ for the natural logarithm,
and we define the \emph{hypersimplex} as $\Delta_{n, k} \defeq \{\mat{x} \in \R^n : \sum_{i=1}^n x_i = k \text{~and~} 0 \le x_i \le 1\}$.

\paragraph{Graphs.}
For a bipartite graph $G = (L, R, E)$,
we denote left and right nodes by $i \in L$ and $j \in R$, respectively.
Let $n = |L|$ since $L$ is the ground set of our subset selection problem and $m = |E|$.
The neighborhood and degree of node $j \in R$ are $N(j) = \{i \in L : (i,j) \in E\}$ and $\deg(j) = |N(j)|$.

\addtocontents{toc}{\protect\setcounter{tocdepth}{0}}
\subsection{\citet{barman2021tight}'s algorithm}
\addtocontents{toc}{\protect\setcounter{tocdepth}{3}}  

We start by reviewing the relax-and-round framework in \citep{barman2021tight} for concave coverage problems.

\paragraph{LP relaxation.}
First, extend the reward function $\varphi$ to be piecewise linear on the domain $\R_{\ge 0}$:
\begin{equation*}
  \varphi(x)
  :=
  \lambda \varphi(\floor{x})
  +
  (1-\lambda) \varphi(\floor{x} + 1),
\end{equation*}
where $\lambda = 1 - (x - \floor{x}) \in [0, 1]$.
This piecewise linear extension is also concave and nondecreasing.
Now consider the natural linear programming (LP) relaxation:
\begin{tcolorbox}[
   title={Concave coverage LP relaxation.},
   colback=white,
]
\vspace{-0.3cm}
\begin{align*}
    \text{maximize~~} & \sum_{j \in R} w_j c_j \\
    \text{subject to~~}
    & y_j = \sum_{i \in N(j)} x_i, ~~ \forall j \in R \\
    & c_j \le \varphi(y_j), ~~ \forall j \in R \\
    & \sum_{i = 1}^n x_i = k \\
    & 0 \le x_i \le 1, ~~ \forall i \in L
\end{align*}
\end{tcolorbox}

\noindent
Each piecewise linear constraint $c_j \le \varphi(y_j)$ can be represented using at most $\deg(j)$ linear constraints since
$\varphi(x)$ is a concave hull and the maximum value of $y_j$ is $\deg(j)$.
For example, if $\deg(j) = n$, then $c_j \le \varphi(y_j)$ is equivalent to
\begin{align*}
    c_j &\le \varphi(0) + (\varphi(1) - \varphi(0)) \cdot (y_j - 0), \\
    c_j &\le \varphi(1) + (\varphi(2) - \varphi(1)) \cdot (y_j - 1), \\
    &~~~~~~~~~~~~~~~~~~~~~~~~~\vdots \\
    c_j &\le \varphi(n-1) + (\varphi(n) - \varphi(n-1)) \cdot (y_j -(n-1)).
    \vspace{-0.2cm}
\end{align*}

Let $\mat{x}^* \in [0, 1]^n$ be an optimal solution to the LP relaxation.
By extending $C$ to the domain $[0,1]^n$ as in the LP relaxation,
the two functions agree on integral points, i.e., $C(\mat{1}_S) = C(S)$,
and we have the upper bound
\begin{equation}
\label{eqn:lp_relaxation_bound_1}
  C(\mat{x}^*)
  \ge
  C(\mat{1}_{S^*})
  =
  \OPT,
\end{equation}
where $S^*$ is an optimal solution to the discrete problem.

\vspace{-0.2cm}
\paragraph{Multilinear extension.}
For any set function
$f : 2^{[n]} \rightarrow \R_{\ge 0}$,
its \emph{multilinear extension}
$F : [0,1]^n \rightarrow \R_{\ge 0}$ is defined as
\vspace{-0.3cm}
\begin{equation*}
    F(\mat{x})
    :=
    \sum_{S \subseteq [n]} f(S) \prod_{i \in S} x_i \prod_{i \not\in S} (1 - x_i).
    \vspace{-0.1cm}
\end{equation*}
There is a probabilistic interpretation of the multilinear extension.
For any $\mat{x} \in [0,1]^n$,
let $S_{\mat{x}}$ be a random subset where each $i \in [n]$
is included independently with probability $x_i$, and not included with probability $1-x_i$.
Then, the multilinear extension of $f$ is simply $F(\mat{x}) = \E[f(S_{\mat{x}})]$.

\vspace{-0.2cm}
\paragraph{Pipage rounding.}
After computing the LP solution $\mat{x}^*$,
\citet{barman2021tight} use \emph{pipage rounding} \citep{ageev2004pipage}
to round $\mat{x}^*$ to a feasible integral vector $\mat{x}^\rounded \in \{0,1\}^n$ satisfying
\begin{equation}
\label{eqn:pipage_rounding_guarantee}
  F(\mat{x}^\rounded) \ge F(\mat{x}^*),
\end{equation}
where the underlying set function is $f(S) = C(S)$.
Pipage rounding can be applied here because
\begin{enumerate}
    \item $C(S)$ is submodular $\implies$ $F(\mat{x})$ is convex along any line of direction $\mat{e}_i - \mat{e}_j$ for $i,j \in [n]$;
    \item $F(\mat{x})$ can be evaluated in polynomial time, e.g., using dynamic programming.
\end{enumerate}

\begin{theorem}[{\citet[Theorem 2]{barman2021tight}}]
For any $\mat{x} \in [0,1]^n$, we have
\begin{equation}
\label{eqn:linking_the_extensions}
    F(\mat{x})
    \ge
    \parens*{\min_{i \in [n]} \alpha_{\varphi}(i)} C(\mat{x}),
\end{equation}
where $F(\mat{x})$ is the multilinear extension of $C(S)$
and $C(\mat{x})$ is the LP relaxation objective.
\end{theorem}

This result connects the LP solution $C(\mat{x}^*) \ge \OPT$ to the multilinear extension $F(\mat{x^*})$,
and is where the Poisson concavity ratio $\alpha_{\varphi}$ gets introduced.
The approximation ratio of \citet{barman2021tight}'s relax-and-round method follows by chaining the inequalities in \eqref{eqn:lp_relaxation_bound_1}, \eqref{eqn:linking_the_extensions}, and \eqref{eqn:pipage_rounding_guarantee}.

\paragraph{Running time.}
The LP has
$O(n + |R|)$ variables,
$O(m)$ linear constraints, and
$O(m)$ nonzeros in the constraint matrix.
Using the fast theoretical LP solvers in~\citep{van2020deterministic, cohen2021solving},
the running time to compute $\mat{x}^*$ to $\varepsilon$-relative accuracy is $\tilde{O}((n + |R|)^{\omega} \log(1/\varepsilon))$,
where $\omega \approx 2.371$ is the matrix multiplication exponent~\citep{alman2025more}.
If $\varphi$ saturates early, i.e., $\varphi(\ell) = \varphi(\ell + 1)$,
then we can write the LP with fewer linear constraints and nonzeros.

After finding $\mat{x}^*$,
pipage rounding runs for $O(n)$ steps and evaluates $F(\mat{x})$ twice in each step.
By the linearity of expectation,
\vspace{-0.15cm}
\begin{align*}
    F(\mat{x})
    &=
    \E[C(S_{\mat{x}})]
    =
    \sum_{j \in R} w_j \E\bracks*{\varphi(|S_{\mat{x}} \cap N(j)|)}.
\end{align*}
Each random variable $|S_{\mat{x}} \cap N(j)|$ follows a \emph{Poisson binomial distribution} with parameter $\mat{x}_{N(j)}$.
\citet[Proposition A.8]{barman2021tight} compute each $\E\bracks{\varphi(|S_{\mat{x}} \cap N(j)|)}$ term using a discrete Fourier transform in $O(\deg(j)^2)$ time.
This means we can evaluate $F(\mat{x})$ in $O\parens{\sum_{j \in R} \deg(j)^2} = O\parens{m \Delta}$ time,
where $\Delta = \max_{j \in R} \deg(j)$.
It follows that the running time of pipage rounding is $O(n m \Delta)$,
and the overall running time of the relax-and-round algorithm is given in \Cref{tab:theory_comparison}.

%% file: smooth_proxy_objective.tex
\section{Smooth approximate coverage function}
\label{sec:smooth_proxy_objective}

Building on \citet{barman2021tight}'s approach,
we present our accelerated relax-and-round algorithm for concave coverage problems.
Our first step is to construct a smooth approximation for $C(\mat{x})$.
This allows us to use accelerated projected gradient methods to design a fast approximate LP solver.

\paragraph{Nesterov smoothing for piecewise linear functions.}

Let $f : \R^n \rightarrow \R$ be the piecewise linear concave function
\[
    f(\mat{x}) = \min_{i \in [d]} \parens*{\mat{a}_{i}^\intercal \mat{x} + b_i}.
\]
Since $f(\mat{x})$ is not continuously differentiable, we use Nesterov smoothing with an entropy regularizer.

\vspace{0.2cm}
\noindent
\begin{minipage}[c]{0.50\textwidth}%
\begin{definition}
For any $\mu > 0$, the \emph{log-sum-exp (LSE) smoothing} of $f(\mat{x})$ is
\[
    f_{\mu}(\mat{x}) \defeq -\mu \log\parens*{\sum_{i=1}^d \exp\parens*{- \frac{\mat{a}_{i}^\intercal \mat{x} + b_i}{\mu}}}.
\]
\end{definition}%
As $\mu \rightarrow 0$, $f_{\mu}(\mat{x})$ converges uniformly from below to $f(\mat{x})$.
We give the gradient, smoothness, and approximation guarantees of $f_\mu(\mat{x})$
in \Cref{app:log_exp_sum_trick}.
\end{minipage}\hfill
\begin{minipage}[c]{0.49\textwidth}
    \centering
    \vspace{-0.05cm}
    \includegraphics[width=0.75\linewidth]{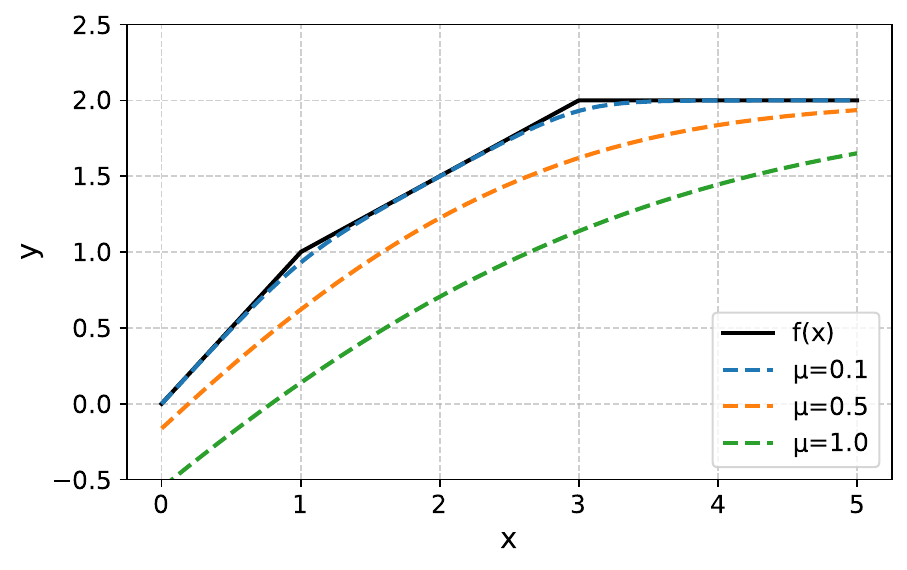}
    \vspace{-0.32cm}
    \captionof{figure}{Examples of LSE smoothing.}
\end{minipage}

\paragraph{Smoothing the coverage function $C(\mat{x})$.}
We apply LSE smoothing to each $j \in R$.
Start by writing
\[
    \varphi_{j}(\mat{x})
    \defeq
    \min_{i = 1}^{\deg(j)} \parens*{\mat{a}_{j,i}^\intercal \mat{x} + b_{i}},
\]
where
\begin{align*}
    \mat{a}_{j, i}
    &=
    (\varphi(i) - \varphi(i - 1)) \mat{1}_{N(j)} \\
    b_{i}
    &=
    \varphi(i - 1) - (\varphi(i) - \varphi(i - 1)) (i - 1).
\end{align*}

\begin{definition}
For any $\mu > 0$, define the \emph{smooth coverage function} as
\begin{equation*}
\label{def:smooth_c}
    \SmoothC(\mat{x})
    \defeq
    \sum_{j \in R} w_j \varphi_{j,\mu}(\mat{x}).
\end{equation*}
Note that $\SmoothC(\mat{x})$ is not the same as applying LSE smoothing to $C(\mat{x})$, hence the tilde notation.
\end{definition}

LSE smoothing is a temperature-controlled softmin, so define the probabilities
\begin{align}
\label{eqn:p_def}
    p_{j, i}(\mat{x})
    &\defeq
    \frac{\exp\parens*{-(\mat{a}_{j, i}^\intercal \mat{x} + b_{i})/\mu}}{\sum_{k=1}^{\deg(j)}\exp\parens{-(\mat{a}_{j, k}^\intercal \mat{x} + b_{k})/\mu}}.
\end{align}

\begin{restatable}{lemma}{SmoothCGaurantees}
\label{lem:smooth_c_guarantees}
The smooth coverage function $\SmoothC(\mat{x})$ has the following properties:
\begin{itemize}
    \item \emph{Gradient.}
    \begin{equation*}
    \label{eqn:smooth_c_grad_def}
        \nabla \SmoothC(\mat{x})
        =
        \sum_{j \in R} w_j \mat{1}_{N(j)} \sum_{i=1}^{\deg(j)} p_{j, i}(\mat{x}) \cdot (\varphi(i) - \varphi(i - 1)).
    \end{equation*}
    \item \emph{Smoothness.}
    Let $\degreeR \defeq \sum_{j \in R} w_j \deg(j)$ be the effective degree in $R$. Then,
    \[
        \norm{\nabla^2 \SmoothC(\mat{x})}_{2} \le \frac{\degreeR}{4 \mu}.
    \]
    \item \emph{Approximation.}
    Let $\mat{x}^*$ be a maximizer of $C(\mat{x})$ and $\widetilde{\mat{x}}_\mu^*$ be a maximizer of $\SmoothC(\mat{x})$. Then,
    \[
        C(\mat{x}^*) - \mu \log \degreeR
        \le
        \widetilde{C}_{\mu}(\widetilde{\mat{x}}^*_\mu)
        \le
        C(\mat{x}^*).
    \]
\end{itemize}
\end{restatable}

\begin{proof}[Proof\xspace sketch]
We give the full proof in \Cref{app:smooth_approximate_coverage_function}.
The gradient $\nabla \SmoothC(\mat{x})$ and the approximation guarantees follow from linearity and LSE smoothing properties.
To bound the smoothness, we relate the Hessian $\nabla^2 \SmoothC(\mat{x})$
to the covariance matrix of the random vector $Y_j(\mat{x}) \mat{1}_{N(j)}$, where $Y_j(\mat{x})$ is a discrete random variable over the slopes $\varphi(i) - \varphi(i - 1)$ with probability mass function $p_j(\mat{x})$.
The result follows from applying Popoviciu’s inequality.
\end{proof}

%% file: algorithm.tex
\section{Algorithm}
\label{sec:algorithm}

We now present our accelerated relax-and-round algorithm.
It starts by running the greedy algorithm to get $\xgreedy \in \{0,1\}^n$, which we show is a $(1-\sfrac{1}{e})$-approximation to the fractional optimum $C(\mat{x}^*)$.
Starting from $\mat{x}^{(0)} = \xgreedy$, it then runs
accelerated projected gradient ascent to optimize
the smooth coverage function $\SmoothC(\mat{x})$
to get $\mat{x}^{(T)} \approx \widetilde{\mat{x}}_{\mu}^*$.
Finally, it rounds the fractional solution $\mat{x}^{(T)} \in [0,1]^n$
using a Carath\'eodory decomposition algorithm \citep{karalias2025geometric} and \SwapRound~\citep{chekuri2010dependent}.
For uniform matroids, this combined rounding method is faster than pipage rounding.
The rest of this section analyzes the approximation ratio and running time of the algorithm.

\begin{algorithm}[H]
\caption{Accelerated relax-and-round algorithm for concave coverage problems.}
\label{alg:accelerated_relax_and_round}
\begin{algorithmic}[1] 
\Procedure{\MainAlgorithm}{$G = (L, R, E)$, weights $w_j$, cardinality constraint $k$, concave reward function $\varphi$, error $\varepsilon > 0$}
    \State $S_0 \gets \varnothing$        \hfill \Comment{Greedy initialization}
    \For{$t = 0, 1, \dots, k - 1$}
        \State $i^*_{t+1} \gets \argmax_{i \in L \setminus S_t} C(S_t \cup \set{i}) - C(S_t)$
        \State $S_{t+1} \gets S_t \cup \{i^*_{t+1}\}$
    \EndFor
    \State $\mat{x}_{\texttt{greedy}} \gets \mat{1}_{S_k}$
    \State $\mu \gets \frac{\varepsilon C(\mat{x}_{\texttt{greedy}})}{2 \log \degreeR}$  \hfill \Comment{Smoothing parameter}
    \State $\eta \gets \frac{4 \mu}{\degreeR}$
    \State $T \gets \ceil*{\frac{2}{\varepsilon C(\xgreedy)} \sqrt{\frac{k \degreeR \log \degreeR}{(1 - \sfrac{1}{e})(1 + \sfrac{1}{e})}}~}$
    \State $\mat{x}^{(0)} \gets \mat{x}_{\texttt{greedy}}$
    \State $\mat{y}^{(0)} \gets \mat{x}^{(0)}$
    \State $\beta_0 \gets 1$
    \For{$t = 0, 1, \dots, T - 1$}
        \State $\mat{x}^{(t+1)} \gets \Pi_{\Delta_{n,k}}(\mat{y}^{(t)} + \eta \nabla \SmoothC(\mat{y}^{(t)}))$ \hfill \Comment{Projected gradient step onto $\Delta_{n,k}$}
        \State $\beta_{t+1} \gets \frac{1 + \sqrt{1 + 4 \beta_t^2}}{2}$
        \State $\mat{y}^{(t+1)} \gets \mat{x}^{(t+1)} + \parens{\frac{\beta_t - 1}{\beta_{t+1}}} (\mat{x}^{(t+1)} - \mat{x}^{(t)})$
    \EndFor
    \State \Return $\SwapRound(\mat{x}^{(T)})$
\EndProcedure
\end{algorithmic}
\end{algorithm}

\begin{theorem}
\label{thm:main_theorem}
For any $\varepsilon > 0$, \MainAlgorithm returns a set $S \subseteq L$ such that $|S| \le k$ and
\[
    \E[C(S)] \ge \parens*{\alpha_\varphi - \varepsilon} \cdot \OPT.
\]
The running time is $O\parens*{(m + n \log n) \frac{n \sqrt{\log n}}{\varepsilon}}$.
\end{theorem}

\subsection{Greedy initialization}

To analyze this algorithm,
we start by extending the analysis of the greedy algorithm in \citet{nemhauser1978analysis}
for cardinality-constrained monotone submodular maximization
from the discrete setting, $\max_{S \subseteq L : |S| \le k} C(S)$,
to the continuous LP relaxation, $\max_{\mat{x} \in \Delta_{n,k}} C(\mat{x})$.
This ensures that $\xgreedy$ achieves a constant-factor approximation
to $C(\mat{x}^*)$,
which we need to compute an optimal smoothing parameter $\mu$.

\begin{restatable}{lemma}{XGreedyLemma}
\label{lem:xgreedy}
The greedy solution $\mat{x}_{\texttt{greedy}} \in \{0,1\}^n$ can be computed in $O(mk)$ time
and satisfies
\[
    C(\mat{x}_{\texttt{greedy}}) \ge \parens*{1 - \sfrac{1}{e}} \cdot C(\mat{x}^*).
\]
\end{restatable}

\begin{proof}[Proof\xspace sketch]
Let $\mat{x}_{S_t} = \mat{1}_{S_t}$ denote the point at step $t$.
We derive a supergradient inequality using the concavity of $C(\mat{x})$ to get
\[
    C(\mat{x}^*) - C(\mat{x}_{S_t}) \le k \cdot (C(\mat{x}_{S_{t+1}}) - C(\mat{x}_{S_t})).
\]
Solving the recurrence gives the classic bound
\[
    C(\mat{x}_{S_{k}}) \ge \parens{1 - (1-\sfrac{1}{k})^k} C(\mat{x}^*).
\]
See \Cref{app:xgreedy} for the full proof.
\end{proof}

\subsection{Accelerated projected gradient method}

Next we list the key ingredients to approximately optimize the smooth problem $\max_{\mat{x} \in \Delta_{n,k}} \SmoothC(\mat{x})$.

\begin{lemma}
\label{lem:compute_gradient}
The gradient $\nabla \SmoothC(\mat{x})$ can be computed in $O(m)$ time.
\end{lemma}

\begin{proof}
We evaluate the gradient in three steps:
compute all of the probabilities $p_{j,i}(\mat{x})$ defined in \eqref{eqn:p_def},
calculate the expectations $\E[Y_{j}(\mat{x})]$ in \Cref{lem:smooth_c_guarantees},
and assemble the final vector.
Each of these steps takes $O(\sum_{j \in R} \deg(j)) = O(m)$ time, so the overall time complexity is $O(m)$.
\end{proof}

\begin{restatable}{lemma}{HypersimplexProjection}
\label{lem:hypersimplex_projection}
The projection of any point $\mat{x} \in \R^n$ onto the hypersimplex $\Delta_{n,k}$,
defined as $\Pi_{\Delta_{n,k}}(\mat{x}) \defeq \argmin_{\mat{y} \in \Delta_{n,k}} \norm{\mat{y} - \mat{x}}_{2}$,
can be computed in $O(n \log n)$ time.
\end{restatable}

\begin{proof}[Proof\xspace sketch]
Using Lagrange multipliers, the solution is of the form $y^*_{i} = \min(1, \max(0, x_i - \lambda))$, for some $\lambda \in \R$.
We use binary search and linear interpolation to find the $\lambda^*$ such that $\sum_{i=1}^n y_i^* = k$.
We defer the algorithm and proof to \Cref{app:hypersimplex_projection}.
\end{proof}

\begin{restatable}[Convergence rate]{lemma}{ConvergenceRate}
\label{lem:convergence_rate}
Let $\mu > 0$ and assume the step size satisfies $\eta \le \frac{4\mu}{\degreeR}$.
Then, for any iteration $t \ge 1$,
\begin{align}
\label{eqn:main_error_bound}
    C(\mat{x}^*) - \widetilde{C}_{\mu}(\mat{x}^{(t)})
    &\le
    \frac{k \degreeR}{\mu (t+1)^2} + \mu \log \degreeR.
\end{align}
\end{restatable}

\begin{proof}[Proof\xspace  sketch]
We combine the approximation guarantee in \Cref{lem:smooth_c_guarantees} for $\SmoothC(\mat{x})$
and the convergence rate analysis of FISTA \citep[Theorem 4.4]{beck2009fast}.
Specifically, we take the penalty term in the proximal operator of FISTA to be the indicator function for the feasible set $\Delta_{n,k}$.
See \Cref{app:accelerated_projected_gradient_ascent} for the full proof.
\end{proof}

Lastly, we give a relative error guarantee by setting the right-hand side of \eqref{eqn:main_error_bound} to $\varepsilon C(\mat{x}^*)$
and optimizing the smoothing parameter $\mu$ to minimize the total iteration count $T$.
This achieves the optimal balance between the LSE approximation error and the FISTA convergence rate.

\begin{restatable}{theorem}{OptimizingMu}
\label{thm:relative_error_algorithm}
Let $\varepsilon > 0$ and $\mu = \frac{\varepsilon C(\xgreedy)}{2 \log \degreeR}$.
If $t \ge \frac{5.1}{\varepsilon} n \sqrt{\log n} $, then
$
    \widetilde{C}_{\mu}(\mat{x}^{(t)}) \ge \parens*{1 - \varepsilon} \cdot C(\mat{x}^*).
$
\end{restatable}

\begin{proof}[Proof sketch]
Ideally, we would use smoothing parameter $\mu^* = \varepsilon C(\mat{x}^*) / (2 \log \degreeR)$,
but this relies on the value we want to compute.
Thus, we use the underestimate $\mu = \varepsilon C(\xgreedy) / (2 \log \degreeR) = \gamma \mu^*$,
which has multiplicative error $\gamma \ge 1-\sfrac{1}{e}$ by \Cref{lem:xgreedy}.
The result follows from the convergence rate in \Cref{lem:convergence_rate}
and applying the lower bound $C(\mat{x}^*) \ge 0.43\sqrt{k \degreeR} / n$ in \Cref{lem:bounding_fractional_opt}.
The value of $T$ in Line~9 of \Cref{alg:accelerated_relax_and_round} is tighter than this result, but this result is a cleaner running time bound.
We defer the full proof to \Cref{app:optimizing_mu}.
\end{proof}

\subsection{Rounding}
In the last stage of our algorithm,
we round the fractional solution $\mat{x}^{(T)} \in [0, 1]^n$
to a feasible integral solution $\mat{x}^{\rounded} \in \{0, 1\}^n$.
\citet{barman2021tight} use pipage rounding~\citep{ageev2004pipage},
but pipage rounding requires $O(nm \Delta)$ time since it
evaluates the multilinear extension $F(\mat{x})$ at most $n$ times,
which can be prohibitively expensive.

To make this faster, we give a tailored version of the \SwapRound algorithm~\citep{chekuri2010dependent}
for uniform matroids (cardinality constraint $k$)
that removes the dependence on the number of edges~$m$.
A subtle but crucial requirement for \SwapRound is that
the input fractional point needs to be represented as a convex combination of matroid bases,
i.e., $\mat{x}^{(T)}$ needs to be written as a convex combination of vertices of the hypersimplex $\Delta_{n,k}$.
Carathéodory's theorem guarantees that such a decomposition exists with at most $n$ vertices,
so we use the recent combinatorial vertex-peeling algorithm of \citet{karalias2025geometric}
for uniform matroids to find this decomposition.

\begin{restatable}[\SwapRound for uniform matroids]{lemma}{HypersimplexSwapRound}
\label{thm:swapround}
Given any point $\mat{x} \in \Delta_{n,k}$ in the basis polytope of the uniform matroid $U_{n}^k$,
there is an $O(nk \log n)$-time implementation of \SwapRound that guarantees
\[
  \E[C(\SwapRound(\mat{x}))] \ge F(\mat{x}).
\]
\end{restatable}

\begin{proof}[Proof sketch]
We find a Carathéodory decomposition of $\mat{x}^{(T)}$
with $r \le n$ bases in $O(rk \log n)$ time using \citet[Theorem 4.3]{karalias2025geometric}.
Running \SwapRound takes $O(rk \log n)$ time for uniform matroids,
which proves the running time.
The expected approximation follows from \citet[Theorem 2.1]{chekuri2010dependent}.
See \Cref{app:rounding} for full details and pseudocode for the subroutines.
\end{proof}

\subsection{Analysis: Putting everything together}

\begin{proof}[Proof of \Cref{thm:main_theorem}]
Our analysis builds on the framework of \citet{barman2021tight}.

\paragraph{Approximation ratio.}
Let $S = \SwapRound (\mat{x}^{(T)})$.
The result follows from a chain of inequalities:
\begin{align*}
\label{eqn:chain_of_inequalities_2}
    \E[\ALG]
    &=
    \E[C(\mat{1}_{S})] \\
    &=
    \E[F(\mat{1}_{S})] & \text{($C(\mat{x}) = F(\mat{x})$ for $\mat{x} \in \{0, 1\}^n$)} \\
    &\ge
    F(\mat{x}^{(T)}) & \text{(\Cref{thm:swapround})} \\
    &\ge
    \alpha_{\varphi} \cdot C(\mat{x}^{(T)}) & (\text{{\citep[Theorem 2]{barman2021tight}}}) \\
    &>
    \alpha_{\varphi} \cdot \SmoothC (\mat{x}^{(T)}) & \text{($C(\mat{x}) > \SmoothC(\mat{x})$ for $\mat{x} \in \R^n_{\ge 0}$)}\\
    &\ge
    (\alpha_{\varphi} - \varepsilon) \cdot C(\mat{x}^*) & \text{(\Cref{thm:relative_error_algorithm})}\\
    &\ge
    (\alpha_{\varphi} - \varepsilon) \cdot C(\mat{1}_{S^*}) \\
    &=
    (\alpha_{\varphi} - \varepsilon) \cdot \OPT.
\end{align*}

\paragraph{Running time.}
\Cref{alg:accelerated_relax_and_round} runs in three stages: greedy initialization, projected gradient ascent, and rounding.
Computing $\mat{x}^{(0)} = \xgreedy$ takes $O(mk)$ time by \Cref{lem:xgreedy},
and evaluating $C(\xgreedy)$ to set $\mu$, $\eta$, and $T$ takes $O(m)$ time.

In each step $t$ of accelerated projected gradient ascent,
computing the gradient of the smooth objective $\nabla \SmoothC(\mat{y}^{(t)})$ takes $O(m)$ time by \Cref{lem:compute_gradient}.
Projecting the gradient step onto the hypersimplex $\Delta_{n,k}$ to get the next point $\mat{x}^{(t+1)}$ takes $O(n \log n)$ time by \Cref{lem:hypersimplex_projection}.
We need $T$ steps to achieve the relative error bound,
so it follows from \Cref{thm:relative_error_algorithm} that the running time of this stage is
\[
    O\parens*{\parens*{m + n \log n} \cdot \frac{n}{\varepsilon} \sqrt{\log n}}.
\]
Finally, once we have computed $\mat{x}^{(T)} \in [0,1]^n$,
we apply the \SwapRound analysis in \Cref{thm:swapround}.
\end{proof}

%% file: approximation_ratios.tex
\section{Approximation ratios}
\label{sec:approximation_ratios}

In this section, we compute the Poisson concavity ratio (and hence approximation ratio of \Cref{alg:accelerated_relax_and_round}) for new reward functions.
We defer the proofs to \Cref{app:approximation_ratios}.

\begin{restatable}{theorem}{LogRatio}
\label{thm:log_ratio}
If $\varphi(x) = \log(1 + x)$, then $\alpha_{\varphi} \approx 0.827$.
\end{restatable}

\begin{proof}[Proof sketch]
Using the formula for $\alpha_{\varphi}(x)$ in \eqref{eqn:poisson_concavity_ratio},
we show that $\alpha_{\varphi}(x) \ge 1 - \frac{1}{x \log(1 + x)}$.
For $x \in \Z_{\le 4}$, we check $\alpha_{\varphi}(x)$ numerically.
For $x \ge 4$, this lower bound lets us prove
$\alpha_{\varphi}(x) \ge \alpha_{\varphi}(1) \approx 0.827$.
\end{proof}

The next two results generalize the maximum coverage problem and recover the $1-1/e \approx 0.632$ approximation ratio.
First, we extend the $c$-multi-coverage objective $\varphi(x) = \min(x, c)$ to a piecewise linear reward that discounts any marginal reward beyond $c$ by a factor of $\beta$.
This perfectly interpolates between the approximation ratio for $c$-multi-coverage functions in \citet[Theorem 1.1]{barman2022tight} and the identity function $\varphi(x) = x$,
which has Poisson concavity ratio $\E_{X \sim \Pois(x)}[X] / x = 1$.

\begin{restatable}{theorem}{PiecewiseLinear}
\label{thm:piecewise_linear_ratio}
Let $c \in \Z_{\ge 1}$, $\beta \in [0, 1]$.
If $\varphi(x) = \beta x + (1 - \beta) \min(x, c)$, then
\[
    \alpha_\varphi = 1 - (1 - \beta) \frac{c^c e^{-c}}{c!}.
\]
\end{restatable}

Finally, we give a closed-form formula for the approximation ratio of isoelastic rewards $\varphi(x) = x^{1 - \gamma}$.
We build on \citet[Proposition B.3]{barman2021tight}, which states
$\alpha_\varphi = \frac{1}{e} \sum_{n=1}^\infty \frac{n^{1-\gamma}}{n!}$.
For example, if $\gamma = 1/2$, then $\varphi(x) = \sqrt{x}$ and $\alpha_{\varphi} \approx 0.773$.
As $\gamma \rightarrow 1$, we recover the maximum coverage problem.

\begin{restatable}{theorem}{Isoelastic}
\label{thm:isoelastic_ratio}
For $\gamma \in (0, 1)$ and $\varphi(x) = x^{1 - \gamma}$, we have
\[
    \alpha_{\varphi} = \frac{1}{e \Gamma(\gamma)} \int_{0}^{1} \frac{e^x}{(-\log x)^{1 - \gamma}} dx.
\]
\end{restatable}

%% file: experiments.tex
\section{Experiments}
\label{sec:experiments}

Finally, we compare relax-and-round methods for $c$-multi-coverage rewards $\varphi(x) = \min(x, c)$.

\subsection{Warm-up: Hard synthetic graphs}

\begin{minipage}[c]{0.51\textwidth}

Inspired by the $(1-\sfrac{1}{e}+\varepsilon)$-hardness of approximation in \citet{feige1998threshold},
we construct a family of instances with the following property in \Cref{app:hard_greedy_instances}.

\begin{restatable}{theorem}{HardInstances}
\label{thm:hard_instances}
For any $\varepsilon > 0$ and $c \ge 1 + \ceil*{3.3 / \varepsilon}$,
there is a $c$-multi-coverage input where the greedy algorithm outputs $S$
with $f(S) \le (1 - \sfrac{1}{e} + \varepsilon) \OPT$.
\end{restatable}

\citet{barman2022tight} prove that $\alpha_{\varphi} = 1 - \frac{c^c e^{-c}}{c!} \rightarrow 1$ as $c$ increases.
\Cref{fig:synthetic_gap} shows a clear gap between the approximation guarantee of \Cref{alg:accelerated_relax_and_round} (dotted blue) and the greedy objective value (solid orange).

\end{minipage}
\hspace{0.2cm}
\begin{minipage}[c]{0.45\textwidth}
    \centering
    \vspace{-0.58cm}
    \includegraphics[width=0.95\linewidth]{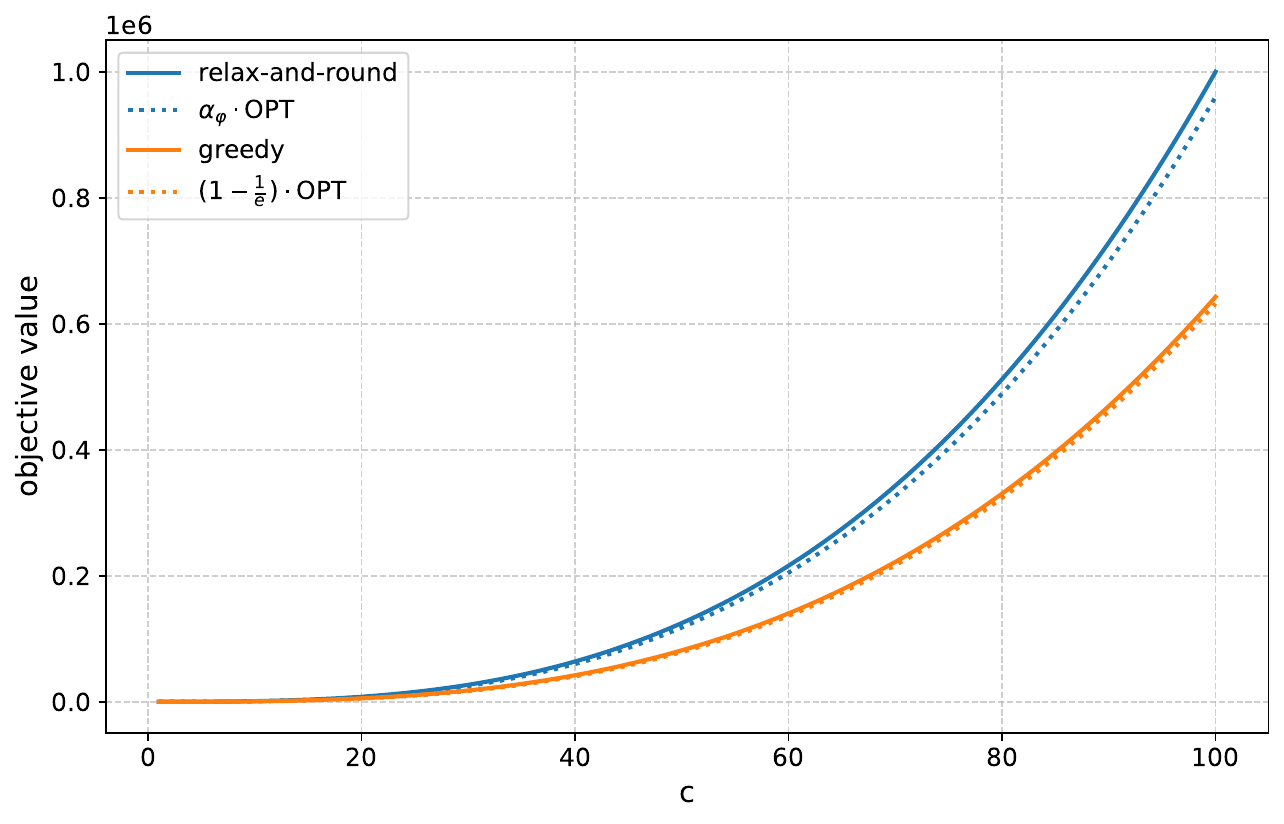}
    \captionof{figure}{Comparison of relax-and-round and greedy algorithms on hard synthetic graphs.}
    \label{fig:synthetic_gap}
\end{minipage}

\subsection{SNAP graphs}

\paragraph{Datasets.}
We build $c$-multi-coverage instances from SNAP graphs~\citep{snapnets}.
For a graph $G = (V, E)$,
let $G' = (L, R, E')$ where
$L = R = V$ and $E' = \{(i,i) : i \in V\} \cup \{(i, j) : \{i,j\} \in E \text{~or~} \{j,i\} \in E\}$.
The \texttt{facebook} graph has $n=\text{4,039}$ and $m=\text{180,507}$;
\texttt{dblp} has $n = \text{317,080}$ and $m = \text{2,416,812}$.
We study $c$-multi-coverage on these bipartite graphs for different values of $c$ and budgets~$k$.

\vspace{-0.1cm}
\paragraph{Setup.}
We use several different LP solvers to compute $\mat{x}^*$:
GLOP is a simplex method in Google OR-Tools~\citep{ortools},
HiGHS is a barrier method~\citep{huangfu2018parallelizing},
and SCIP is an optimization library for mixed-integer programming~\citep{hojny2025scip}.
We set the time limit for each LP solve to 600 seconds.
If the LP solver does not find a feasible solution, we report ``TLE'' (time limit exceeded);
otherwise, we run pipage rounding as in \citep{barman2021tight} and report the total running time.

We compare the LP-based relax-and-round methods to \MainAlgorithm (\Cref{alg:accelerated_relax_and_round}) with $\varepsilon = 0.01$.
We vary the step size $\eta$ as a hyperparameter
and apply an early stopping rule if the relative improvement of the smooth objective drops below $\texttt{tol} = \text{1e-6}$.
The initial solution $\mat{x}^{(0)}$ is computed with the lazy greedy algorithm.
All of our experiments are implemented in C\texttt{++}.
We give machine details and LP solver versions in \Cref{app:experiments}.

\vspace{-0.25cm}
\begin{table}[H]
  \caption{End-to-end running times of relax-and-round methods in seconds for $c=2$ multi-coverage.}
  \label{table:main_experiments}
  \centering
  \begin{tabular}{rrrrrrrrrrrrrrrrrrr}
    \toprule
    & & \multicolumn{3}{c}{\textbf{LP Solver + PipageRounding}} & \multicolumn{3}{c}{\textbf{AcceleratedRelaxAndRound}} \\
    \cmidrule(r){3-5} \cmidrule(l){6-8}
    & $k$ & GLOP & HiGHS & SCIP & $\eta = 10$ & $\eta = 1$ & $\eta = 0.1$ \\
    \midrule
    \multirow{5}{*}{\rotatebox[origin=c]{90}{\texttt{facebook}}} &
    20 & 0.84 \scriptsize{$\pm$ 0.01} & 1.12 \scriptsize{$\pm$ 0.01} & 7.64 \scriptsize{$\pm$ 0.14} & \textbf{0.04} \scriptsize{$\pm$ 0.00} & 0.05 \scriptsize{$\pm$ 0.00} & 0.98 \scriptsize{$\pm$ 0.01} \\
    & 40 & 4.00 \scriptsize{$\pm$ 0.01} & 4.34 \scriptsize{$\pm$ 0.02} & 10.60 \scriptsize{$\pm$ 0.47} & \textbf{0.04} \scriptsize{$\pm$ 0.00} & 0.08 \scriptsize{$\pm$ 0.00} & 1.46 \scriptsize{$\pm$ 0.00} \\
    & 60 & 6.76 \scriptsize{$\pm$ 0.06} & 7.02 \scriptsize{$\pm$ 0.04} & 13.66 \scriptsize{$\pm$ 1.19} & \textbf{0.04} \scriptsize{$\pm$ 0.00} & 0.07 \scriptsize{$\pm$ 0.00} & 2.08 \scriptsize{$\pm$ 0.01} \\
    & 80 & 16.58 \scriptsize{$\pm$ 0.11} & 16.56 \scriptsize{$\pm$ 0.03} & 22.47 \scriptsize{$\pm$ 2.08} & \textbf{0.05} \scriptsize{$\pm$ 0.00} & 0.09 \scriptsize{$\pm$ 0.00} & 3.81 \scriptsize{$\pm$ 0.02} \\
    & 100 & 20.28 \scriptsize{$\pm$ 0.14} & 20.74 \scriptsize{$\pm$ 1.21} & 24.97 \scriptsize{$\pm$ 0.71} & \textbf{0.05} \scriptsize{$\pm$ 0.00} & 0.14 \scriptsize{$\pm$ 0.00} & 3.95 \scriptsize{$\pm$ 0.02} \\
    \midrule
    \multirow{5}{*}{\rotatebox[origin=c]{90}{\texttt{dblp}}}
    & 20 & 105.79 \scriptsize{$\pm$ 0.38} & 50.45 \scriptsize{$\pm$ 0.63} & TLE & 1.42 \scriptsize{$\pm$ 0.05} & \textbf{1.38} \scriptsize{$\pm$ 0.05} & 2.09 \scriptsize{$\pm$ 0.09} \\
    & 40 & 116.98 \scriptsize{$\pm$ 0.28} & 56.36 \scriptsize{$\pm$ 1.39} & TLE & 1.58 \scriptsize{$\pm$ 0.02} & 1.79 \scriptsize{$\pm$ 0.04} & \textbf{1.48} \scriptsize{$\pm$ 0.04} \\
    & 60 & 665.57 \scriptsize{$\pm$ 0.39} & 78.01 \scriptsize{$\pm$ 1.68} & TLE & 1.53 \scriptsize{$\pm$ 0.04} & 1.80 \scriptsize{$\pm$ 0.05} & \textbf{1.51} \scriptsize{$\pm$ 0.02} \\
    & 80 & 733.57 \scriptsize{$\pm$ 0.98} & 75.13 \scriptsize{$\pm$ 1.55} & TLE & 1.53 \scriptsize{$\pm$ 0.00} & 2.24 \scriptsize{$\pm$ 0.07} & \textbf{1.48} \scriptsize{$\pm$ 0.01} \\
    & 100 & 732.11 \scriptsize{$\pm$ 1.40} & 86.70 \scriptsize{$\pm$ 1.53} & TLE & 1.56 \scriptsize{$\pm$ 0.02} & 2.28 \scriptsize{$\pm$ 0.02} & \textbf{1.47} \scriptsize{$\pm$ 0.01} \\
    \bottomrule
  \end{tabular}
\end{table}
\paragraph{Results.}
\Cref{table:main_experiments} gives end-to-end running times in seconds averaged across 3 trials.
We focus on $c = 2$, but sweep over $c \in \{1, 2, 4\}$ and report all objective values in \Cref{app:experiments}.
Since the LP solver limit is 600s, pipage rounding on \dblp using GLOP ($k \ge 60$) takes over 130s, which is \mytilde 100x slower than \Cref{alg:accelerated_relax_and_round}.
SCIP does not find a feasible solution on \dblp for any~$k$.
GLOP exhausts its time limit for $k \ge 60$ and only finds suboptimal solutions.
HiGHS finds optimal solutions for $c = 2$,
but \emph{fails to find any feasible solutions in one hour} on \dblp if $c = 1$;
in contrast,
\Cref{alg:accelerated_relax_and_round} runs in 180s for $k \le 10^4$ and finds better solutions than greedy (see \Cref{table:dblp_objective_values_large_k_c1} and \Cref{table:dblp_running_times_large_k_c1}).

%% file: app_introduction.tex
\section{Analysis for \Cref{sec:introduction}}

\subsection{Poisson concavity ratio lower bound}
\label{app:poisson_concavity_ratio_lower_bound}

\begin{lemma}
For any concave nondecreasing function $\varphi : \Z_{\ge 0} \rightarrow \R$ such that $\varphi(0) = 0$,
\[
    \alpha_{\varphi} \ge 1 - \frac{1}{e}.
\]
\end{lemma}

\begin{proof}
By \citet[Proposition A.1]{barman2021tight},
it suffices to show that $\alpha_{\varphi}(k) \ge 1 - \sfrac{1}{e}$ for all positive integers $k$.
We construct a lower bound $f_k$ for $\varphi$ that holds for any choice of $k$.
For $x \le k$, by the definition of concavity,
\[
    \varphi(x) \ge \varphi(k) \frac{x}{k}.
\]
For $x \ge k$, since $\varphi$ is nondecreasing,
\[
    \varphi(x) \ge \varphi(k).
\]
Combining these two piecewise bounds into a single inequality for $x \ge 0$, let
\[
    f_k(x)
    \defeq
    \frac{\varphi(k)}{k} \min\parens*{x, k}.
\]

Now that we have a lower bound for $\varphi$ that is purely a multiple of $\min(x, k)$,
fix the value of $k \ge 1$ and let $X \sim \Pois(k)$.
It follows that
\begin{align*}
    \alpha_{\varphi}(k)
    &=
    \frac{\E[\varphi(X)]}{\varphi(k)}
    \ge
    \frac{\E\bracks*{\min\parens*{X, k}}}{k}.
\end{align*}
Applying \Cref{lem:expected_min},
\begin{align*}
    \E\bracks*{\min\parens*{X, k}}
    &=
    k \Pr(X \le k - 2) + k \Pr(X \ge k) \\
    &=
    k \parens*{1 - \Pr\parens*{X = k - 1}}.
\end{align*}
Substituting this into our inequality,
\[
    \alpha_{\varphi}(k)
    \ge
    \frac{k \parens*{1 - \Pr(X = k - 1)}}{k}
    =
    1 - \frac{k^{k-1} e^{-k}}{(k-1)!}
    =
    1 - \frac{k^{k} e^{-k}}{k!}
    =
    1 - \Pr(X = k).
\]
This value is minimized at $k = 1$, which proves the claim.
\end{proof}

\subsection{Comparison of running times in \Cref{tab:theory_comparison}}
\label{app:runtime_comparison}
We prove that the running time of our algorithm is asymptotically faster than \citet{barman2021tight},
even if they use the current-fastest theoretical LP solvers in~\citep{van2020deterministic,cohen2021solving}.

\begin{theorem}
The worst-case running time of Algorithm~\ref{alg:accelerated_relax_and_round} (\MainAlgorithm) is faster than the worst-case running time of \citet{barman2021tight}, i.e.,
\[
    (m + n \log n) \frac{n \sqrt{\log n}}{\varepsilon} = O\parens*{(n + |R|)^{\omega} \log(1/\varepsilon) + n m \Delta},
\]
if $\omega > 2$ and $\varepsilon > 0$ is a constant.
The current-best matrix multiplication constant is $\omega \approx 2.3714$~\citep{alman2025more}.
\end{theorem}

\begin{proof}
We start by lower bounding the running time of \citet{barman2021tight}:
\begin{align*}
    (n + |R|)^{\omega} \log(1/\varepsilon) + n m \Delta &> |R|^{\omega} + n m \Delta \\
    &\geq |R|^{\omega} + n m \frac{m}{|R|} \\ 
    &\geq \parens*{|R|^{\omega} \cdot \parens*{n m \frac{m}{|R|}}^{\omega} }^{1/(\omega + 1)} & \text{(\Cref{lem:arthmetic_geometric_mean} with $c=1$ and $d=\omega$)}\\
    &= n^{\omega/(\omega+1)} m^{2\omega/(\omega+1)} \\
    &= m \cdot n^{\omega/(\omega+1)} m^{(\omega-1)/(\omega+1)} \\
    &\geq m \cdot n \cdot m^{(\omega-2)/(\omega+1)} & \text{(using $m \geq n$)} \\
    &\geq mn \cdot n^{0.11} & \text{(using $\omega > 2.371$)}
\end{align*}
We use the inequality $m \geq n$ above since any left node with degree zero can be removed from the instance without loss of generality.

To complete the proof, we note that our runtime is essentially the $mn$ term, and that the extra logarithm factors are asymptotically smaller than the polynomial term $n^{0.11}$ even though the exponent constant is small.
Concretely,
\begin{align*}
    (m + n \log n) \frac{n \sqrt{\log n}}{\varepsilon} &= O\parens*{mn \log^2 n} \leq O(m n \cdot n^{0.11}),
\end{align*}
as desired.
\end{proof}

\begin{lemma}\label{lem:arthmetic_geometric_mean}
For any four positive real numbers $a, b, c$, and $d$, we have
$
    a+b \ge \parens{a^c \times b^d}^{1/(c+d)}.
$
\end{lemma}

\begin{proof}
Let $z = \max \set{a, b}$.
The claim follows from a sequence of inequalities:
\begin{align*}
    a + b &\geq \max \set{a, b} \\
    &= z \\
    &= \parens*{z^{c+d}}^{1/(c+d)} \\
    &= \parens*{z^c \times z^d}^{1/(c+d)} \\
    &\geq \parens*{a^c \times b^d}^{1/(c+d)}. \qedhere
\end{align*}
\end{proof}

%% file: app_smooth_proxy_objective.tex
\section{Analysis for \Cref{sec:smooth_proxy_objective}}
\label{app:smooth_proxy_objective}

\subsection{Properties of log-sum-exp smoothing}
\label{app:log_exp_sum_trick}

Recall that
\begin{align*}
    f(\mat{x})
    &=
    \min_{i \in [d]} \parens*{\mat{a}_{i}^\intercal \mat{x} + b_i} \\
    f_{\mu}(\mat{x})
    &=
    -\mu \log\parens*{\sum_{i=1}^d \exp\parens*{- \frac{\mat{a}_{i}^\intercal \mat{x} + b_i}{\mu}}}.
\end{align*}

\paragraph{Gradient.}
Following \citet[Example A.2]{boyd2004convex},
\begin{equation}
\label{eqn:lse_gradient}
    \nabla f_{\mu}(\mat{x}) = \sum_{i=1}^d p_{i}(\mat{x}) \mat{a}_{i},
\end{equation}
where $p_{i}(\mat{x})$ is the softmin probability distribution over the pieces:
\[
    p_i(\mat{x})
    =
    \frac{\exp\parens*{-(\mat{a}_{i}^\intercal \mat{x} + b_i)/\mu}}{\sum_{j=1}^d \exp\parens*{-(\mat{a}_{j}^\intercal \mat{x} + b_j)/\mu}}.
\]
Therefore, the gradient $\nabla f_\mu(\mat{x})$ is a convex combination of the $\mat{a}_i$ vectors.

\paragraph{Hessian.}
Similarly, following \citet[Example A.4]{boyd2004convex},
\begin{align}
\label{eqn:lse_hessian}
    \nabla^2 f_\mu(\mat{x})
    =
    -\frac{1}{\mu}\bracks*{\sum_{i=1}^d p_{i}(\mat{x}) \mat{a}_i \mat{a}_i^\intercal - \parens*{\sum_{i=1}^d p_{i}(\mat{x}) \mat{a}_{i} } \parens*{\sum_{i=1}^d p_{i}(\mat{x}) \mat{a}_{i}}^\intercal }.
\end{align}

\paragraph{Smoothness.}
We bound the smoothness constant of $f_\mu$ by bounding the spectral norm of its Hessian
\citep[Lemma 1.2.2]{nesterov2018lectures}:
\[
    \norm*{\nabla^2 f_\mu(\mat{x})}_{2} \le L.
\]

Since $f_\mu(\mat{x})$ is concave,
$\nabla^2 f_\mu(\mat{x})$ is negative semi-definite
and the spectral norm equals the absolute value of the smallest eigenvalue.
By the Courant--Fisher theorem,
\begin{align*}
    \norm*{\nabla^2 f_\mu(\mat{x})}_{2}
    =
    \abs*{\lambda_{\min}(\nabla^2 f_\mu(\mat{x}))}
    =
    \max_{\norm{\mat{v}}_{2} = 1} \abs*{\mat{v}^\intercal \nabla^2 f_\mu(\mat{x}) \mat{v}}.
\end{align*}
For any $\mat{v} \in \R^n$, the quadratic form is proportional to the negative variance of the discrete
random variable $Z$ that takes the value $\mat{v}^\intercal \mat{a}_i$ with probability $p_i(\mat{x})$:
\[
    - \mat{v}^\intercal \nabla^2 f_\mu(\mat{x}) \mat{v} = \frac{1}{\mu} \Var_{p(\mat{x})}(\mat{v}^\intercal \mat{a}_i).
\]
To find the maximum possible variance, we use Popoviciu's inequality,
which states that the variance of any bounded random variable $Z \in [m, M]$ is bounded by $(M - m)^2 / 4$.

For a unit vector $\mat{v}$, the maximum difference between any two values of the random variable $\mat{v}^\intercal \mat{a}_i$ is
\begin{align*}
    \max_{i, j} \parens*{\mat{v}^\intercal \mat{a}_i - \mat{v}^\intercal \mat{a}_j}
    &=
    \max_{i, j} \mat{v}^\intercal \parens*{\mat{a}_i - \mat{a}_j} \\
    &\le
    \max_{i, j} \norm{\mat{a}_i - \mat{a}_j}_2
\end{align*}
by the Cauchy--Schwarz inequality.
Applying Popoviciu's inequality gives
\[
    L = \frac{1}{4 \mu} \max_{i, j} \norm{\mat{a}_i - \mat{a}_j}^2_2.
\]

\paragraph{Approximation.}
Following \citet[Section 3.1.5]{boyd2004convex}, for any $\mu > 0$,
\begin{equation}
\label{eqn:lse_approximation}
    f(\mat{x}) - \mu \log d \le f_{\mu}(\mat{x}) \le f(\mat{x}).
\end{equation}

\subsection{Proof of \Cref{lem:smooth_c_guarantees}}
\label{app:smooth_approximate_coverage_function}

\SmoothCGaurantees*

\begin{proof}
We build on the properties of LSE smoothing for piecewise linear functions in \Cref{app:log_exp_sum_trick}.

\paragraph{Gradient.}
Since $\mat{a}_{j,i} = (\varphi(i)-\varphi(i-1)) \mat{1}_{N(j)}$, by linearity and \eqref{eqn:lse_gradient} we have
\begin{align*}
    \nabla \SmoothC(\mat{x})
    &=
    \sum_{j \in R} w_j \nabla \varphi_{j, \mu}(\mat{x}) \\
    &=
    \sum_{j \in R} w_j \sum_{i=1}^{\deg(j)} p_{j, i}(\mat{x}) \mat{a}_{j, i} \\
    &=
    \sum_{j \in R} w_j \mat{1}_{N(j)} \sum_{i=1}^{\deg(j)} p_{j, i}(\mat{x}) \cdot (\varphi(i) - \varphi(i - 1)).
\end{align*}

\paragraph{Hessian.}
For $j \in R$, let $Y_j(\mat{x})$ be a discrete random variable over the slopes of $\varphi$ with probability mass function $p_{j}(\mat{x})$.
Then,
\begin{align*}
    \E[Y_{j}(\mat{x})]
    &=
    \sum_{i=1}^{\deg(j)} p_{j, i}(\mat{x}) \cdot (\varphi(i) - \varphi(i - 1)) \\
    \Var(Y_j(\mat{x}))
    &=
    \E[Y_{j}(\mat{x})^2] - \E[Y_{j}(\mat{x})]^2.
\end{align*}
Using \eqref{eqn:lse_hessian}, observe that
\begin{align*}
    \nabla^2 \varphi_{j, \mu}(\mat{x})
    &=
    -\frac{1}{\mu}\bracks*{ \sum_{i=1}^{\deg(j)} p_{j,i}(\mat{x}) \mat{a}_{j,i} \mat{a}_{j,i}^\intercal - \parens*{\sum_{i=1}^{\deg(j)} p_{j,i}(\mat{x}) \mat{a}_{j,i} } \parens*{\sum_{i=1}^{\deg(j)} p_{j,i}(\mat{x}) \mat{a}_{j,i}}^\intercal } \\
    &=
    -\frac{1}{\mu}\bracks*{\mat{1}_{N(j)} \mat{1}_{N(j)}^\intercal \E[Y_j(\mat{x})^2] - \mat{1}_{N(j)} \mat{1}_{N(j)}^\intercal \E[Y_j(\mat{x})]^2 } \\
    &=
    -\frac{1}{\mu} \cdot \Var(Y_{j}(\mat{x})) \cdot \mat{1}_{N(j)} \mat{1}_{N(j)}^\intercal.
\end{align*}
Therefore,
\begin{align*}
    \nabla^2 \SmoothC (\mat{x})
    &=
    \sum_{j \in R} w_j \nabla^2 \varphi_{j, \mu}(\mat{x}) \\
    &=
    -\frac{1}{\mu} \sum_{j \in R} w_j \cdot \Var(Y_{j}(\mat{x})) \cdot \mat{1}_{N(j)} \mat{1}_{N(j)}^\intercal.
\end{align*}

\paragraph{Smoothness.}
\hspace{-0.2cm}
We bound the smoothness constant of $\SmoothC$ by bounding the spectral norm of its Hessian
\citep[Lemma 1.2.2]{nesterov2018lectures}.
By the triangle inequality,
\[
    \norm{\nabla^2 \SmoothC(\mat{x})}_{2}
    \le
    \frac{1}{\mu} \sum_{j \in R} w_j \cdot \Var(Y_{j}(\mat{x})) \cdot \norm{\mat{1}_{N(j)} \mat{1}_{N(j)}^\intercal}_{2}.
\]
Applying Popoviciu's inequality for the variance of bounded random variables,
\[
    \Var(Y_j(\mat{x})) \le \frac{(1 - 0)^2}{4} = \frac{1}{4}.
\]
The spectral norm of a rank-1 matrix is equal to the product of the Euclidean norms of the two vectors that define it, so
\[
    \norm{\mat{1}_{N(j)} \mat{1}_{N(j)}^\intercal}_{2}
    =
    \deg(j).
\]
Putting everything together,
\[
    \norm{\nabla^2 \SmoothC(\mat{x})}_{2}
    \le
    \frac{1}{4 \mu} \sum_{j \in R} w_{j} \deg(j)
    =
    \frac{\degreeR}{4 \mu}.
\]

\paragraph{Approximation.}
The approximation guarantee in \eqref{eqn:lse_approximation} implies that for each $j \in R$,
\[
    \varphi_j(\mat{x}) - \mu \log \deg(j) \le \varphi_{j,\mu}(\mat{x}) \le \varphi_{j}(\mat{x}).
\]
Therefore,
\begin{align}
\label{eqn:smooth_c_pointwise_bounds}
    C(\mat{x}) - \mu \sum_{j \in R} w_j \log \deg(j)
    \le
    \SmoothC(\mat{x})
    \le
    C(\mat{x}).
\end{align}
Let $E$ be the maximum pointwise error:
\[
    E := \mu \sum_{j \in R} w_{j} \log \deg(j).
\]
\emph{Upper bound.}
Since $\SmoothC(\mat{x}) \le C(\mat{x})$ for all $\mat{x}$, it must be true that at $\widetilde{\mat{x}}^*_{\mu}$,
\[
    \widetilde{C}_\mu(\widetilde{\mat{x}}^*_{\mu}) \le C(\widetilde{\mat{x}}^*_{\mu}).
\]
Since $\mat{x}^*$ is the global maximizer of $C(\mat{x})$, we also have $C(\widetilde{\mat{x}}^*_{\mu}) \le C(\mat{x}^*)$.
Combining these together gives
\[
    \widetilde{C}_\mu(\widetilde{\mat{x}}_{\mu}^*) \le C(\mat{x}^*).
\]

\emph{Lower bound.}
Since $\widetilde{\mat{x}}^*_{\mu}$ maximizes $\SmoothC(\mat{x})$,
we have $\SmoothC(\mat{x}^*) \le \SmoothC(\widetilde{\mat{x}}^*_{\mu})$.
From our pointwise bounds in \eqref{eqn:smooth_c_pointwise_bounds}, we also know that at $\mat{x}^*$,
\[
    C(\mat{x}^*) - E \le \SmoothC(\mat{x}^*).
\]
Combining these inequalities gives
\[
    C(\mat{x}^*) - E \le \SmoothC(\widetilde{\mat{x}}^*_{\mu}).
\]

\emph{Bounding $E$.}
Since $\log(x)$ is concave, Jensen's inequality gives
\begin{align*}
    E
    =
    \mu \sum_{j \in R} w_j \log \deg(j)
    \le
    \mu \log\parens*{\sum_{j \in R} w_j \deg(j)}
    =
    \mu \log \degreeR,
\end{align*}
which completes the proof.
\end{proof}

%% file: first_order_methods.tex
\section{Analysis for \Cref{sec:algorithm}}

\subsection{Proof of \Cref{lem:xgreedy}}
\label{app:xgreedy}

\XGreedyLemma*

\begin{proof}
The LP relaxation objective $C(\mat{x})$ is concave on $[0,1]^n$,
so we have the following supergradient inequality for any $\mat{x} \in [0, 1]^n$ and $\mat{g} \in \partial^+ C(\mat{x})$:
\begin{equation*}
\label{eqn:supergradient_inequality_1}
    C(\mat{x}^*) - C(\mat{x}) \le \mat{g}^\intercal\parens{\mat{x}^* - \mat{x}}.
\end{equation*}

We construct a supergradient at each point $\mat{x}_{S_{t}} \in \{0,1\}^n$, where $\mat{x}_{S_{t}}$ is the greedy set at step $t$.
For the piecewise linear concave function $\varphi$ evaluated at an integer $x$, the forward difference $\varphi(x + 1) - \varphi(x)$ is a supergradient,
so define $\mat{g} \in \partial^+ C(\mat{x}_{S_t})$ to be
\begin{align*}
    g(\mat{x}_{S_t})_i
    &\defeq
    \sum_{j \in N(i)} w_j \parens*{\varphi(|S_t \cap N(j)| + 1) - \varphi(|S_t \cap N(j)|)}.
\end{align*}
We can write the supergradient inequality as
\begin{align*}
    C(\mat{x}^*) - C(\mat{x}_{S_t})
    &\le
    \sum_{i=1}^n g(\mat{x}_{S_t})_i \cdot (x_i^* - (\mat{x}_{S_t})_i) \\
    &=
    \sum_{i \in S_t} g(\mat{x}_{S_t})_i \cdot (x_i^* - 1)
    +
    \sum_{i \not\in S_t} g(\mat{x}_{S_t})_i \cdot x_i^*.
\end{align*}

Each $x_i^* - 1 \le 0$ since $\mat{x}^* \in [0, 1]^n$.
Further, $\varphi$ is nondecreasing, so the marginal gains $g(\mat{x}_{S_t})_i$ are nonnegative.
It follows that the first sum is nonpositive, leaving us with
\begin{align*}
    C(\mat{x}^*) - C(\mat{x}_{S_t})
    \le
    \sum_{i \not\in S_t} g(\mat{x}_{S_t})_i \cdot x_i^*.
\end{align*}

Observe that for $i \not\in S_t$, we have $g(\mat{x}_{S_t})_i = C(S_t \cup \{i\}) - C(S_t)$.
It will be convenient to define
\[
    \Delta(i \mid S_t) \defeq C(S_t \cup \{i\}) - C(S_t).
\]

Now, consider the greedy choice at step $t$.
The algorithm selects the element $i_{t+1}^*$ with maximum marginal gain.
Therefore, $\Delta(i \mid S_t) \le \Delta(i_{t+1}^* \mid S_t)$ for all $i \in L \setminus S_t$.
Substituting this upper bound,
\begin{align*}
    C(\mat{x}^*) - C(\mat{x}_{S_t})
    \le
    \Delta(i^*_{t+1} \mid S_t) \sum_{i \not\in S_t} x_i^*.
\end{align*}
Since $\mat{x}^*$ is a feasible solution, we know that $\sum_{i=1}^n x_i^* = k$.
It follows that
\begin{align*}
    C(\mat{x}^*) - C(\mat{x}_{S_t})
    \le
    k \cdot \Delta(i^*_{t+1} \mid S_t).
\end{align*}

Further, observe that $\Delta(i^*_{t+1} \mid S_t) = C(\mat{x}_{S_{t+1}}) - C(\mat{x}_{S_t})$.
Substituting this into the inequality above,
\begin{align*}
    C(\mat{x}^*) - C(\mat{x}_{S_t})
    \le
    k \cdot \bracks*{C(\mat{x}_{S_{t+1}}) - C(\mat{x}_{S_t})}.
\end{align*}
Rearranging terms,
\begin{align*}
\label{eqn:supergradient_inequality_7}
    C(\mat{x}_{S_{t+1}}) \ge \frac{1}{k}C(\mat{x}^*) + \parens*{1 - \frac{1}{k}}C(\mat{x}_{S_t}).
\end{align*}
Applying this recurrence $k$ times starting from $C(\mat{x}_{S_0}) = 0$ gives
\begin{align*}
    C(\mat{x}_{S_{k}})
    \ge
    \frac{1}{k}C(\mat{x}^*) + \parens*{1 - \frac{1}{k}}C(\mat{x}_{S_{k-1}})
    \implies
    C(\mat{x}_{S_{k}})
    &\ge
    \frac{1}{k} \sum_{t=0}^{k-1}
    \parens*{1 - \frac{1}{k}}^t C(\mat{x}^*)
    \\
    &=
    \frac{1}{k} \bracks*{\frac{1 - (1-\sfrac{1}{k})^{k}}{1 - (1 - \sfrac{1}{k})}} C(\mat{x}^*) \\
    &=
    \bracks*{1 - \parens*{1 - \frac{1}{k}}^k} C(\mat{x}^*).
\end{align*}
Using the bound $(1 - 1/k)^k < 1/e$, we arrive at the desired result
\[
    C(\xgreedy) \ge \parens*{1 - \frac{1}{e}} C(\mat{x}^*).
\]

\paragraph{Running time.}
We can compute all marginal gains $\Delta(i \mid S_t)$ for $i \in L$ in $O(m)$ time,
so the total running time is $O(mk)$.
\end{proof}

\subsection{Proof of \Cref{lem:hypersimplex_projection}}
\label{app:hypersimplex_projection}

For any $\mat{x} \in \R^n$ and $\lambda \in \R$, define the \emph{clamped shifted sum} as
\[
    F(\lambda ; \mat{x}) \defeq \sum_{i = 1}^n \min\parens*{\max\parens*{x_i - \lambda, 0}, 1}.
\]
The hypersimplex projection algorithm below is similar to \citet[Algorithm 1]{wang2015projection}, but uses binary search instead of a linear scan.
\citet{ang2021fast} also studied this problem and made steps towards a linear-time algorithm.

\begin{algorithm}[H]
\caption{Projection onto the hypersimplex $\Delta_{n,k} = \{\mat{x} \in \R^n : \sum_{i=1}^n x_i = k, 0 \le x_i \le 1 \}$.}
\label{alg:euclidean_projection_hypersimplex}
\begin{algorithmic}[1] 
\Procedure{\texttt{HypersimplexProjection}}{$\mat{x} \in \R^n$, $k$}
    \State $\texttt{points} \gets \bigcup_{i = 1}^n \set{x_i - 1, x_i}$
    \State Sort $\texttt{points}$ in increasing order
    \State Use binary search to find index $i^*$ such that
    \[
        F(\texttt{points}[i] ; \mat{x}) \ge k > F(\texttt{points}[i + 1] ; \mat{x})
    \]
    \State $\lambda_1 \gets \texttt{points}[i^*]$
    \State $\lambda_2 \gets \texttt{points}[i^* + 1]$
    \State $\lambda^* \gets \lambda_1 + (\lambda_2 - \lambda_1) \frac{F(\lambda_1 ; \mat{x}) - k}{F(\lambda_1 ; \mat{x}) - F(\lambda_2 ; \mat{x})}$ 
    \Comment{Linear interpolation}
    \State Initialize $\mat{y} \gets \mat{0}_n$
    \For{$i = 1, 2, \dots, n$}
        \State $y_i \gets \min(\max(x_i -\lambda^*, 0), 1)$
    \EndFor
    \State \Return $\mat{y}$
\EndProcedure
\end{algorithmic}
\end{algorithm}

\HypersimplexProjection*

\begin{proof}
Projecting $\mathbf{x} \in \mathbb{R}^n$ onto the hypersimplex $\Delta_{n,k}$ is equivalent to finding the point $\mathbf{y} \in \Delta_{n,k}$ that minimizes the Euclidean distance to $\mathbf{x}$,
which is equivalent to the following optimization problem: 
\begin{align*}
    &\hspace{2.6cm} \min_{\mat{y} \in \R^n} \frac{1}{2} \norm{\mat{y} - \mat{x}}_2^2 \\
    &\text{subject to} \quad \sum_{i=1}^n y_i = k, \quad \text{and} \quad 0 \le y_i \le 1 \quad \forall i \in [n]
\end{align*}

We first write the Lagrangian for this convex optimization problem
and apply the KKT conditions:
\[
    \mathcal{L}(\mat{y}, \lambda, \boldsymbol{\alpha}, \boldsymbol{\beta} ; \mat{x})
    \defeq
    \frac{1}{2} \sum_{i=1}^n (y_i - x_i)^2 + \lambda \left(\sum_{i=1}^n y_i - k\right) - \sum_{i=1}^n \alpha_i y_i + \sum_{i=1}^n \beta_i (y_i - 1),
\]
where $\lambda \in \mathbb{R}$ is the Lagrange multiplier for the equality constraint, and $\boldsymbol{\alpha} \ge \mathbf{0}, \boldsymbol{\beta} \ge \mathbf{0}$ are multipliers for the inequality constraints $y_i \ge 0$ and $y_i \le 1$.
Setting the partial derivative with respect to $y_i$ to zero gives the stationarity condition:
\[
\frac{\partial \mathcal{L}}{\partial y_i} = y_i - x_i + \lambda - \alpha_i + \beta_i = 0 \implies y_i = x_i - \lambda + \alpha_i - \beta_i 
\]
Using the complementary slackness conditions ($\alpha_i y_i = 0$ and $\beta_i(y_i - 1) = 0$), we can deduce the value of the optimal $y^*_i$:
\begin{itemize}
    \item If $0 < y^*_i < 1$, then $\alpha_i = 0$ and $\beta_i = 0$, giving $y^*_i = x_i - \lambda$.
    
    \item If $y^*_i = 0$, then $\beta_i = 0$ and $\alpha_i \ge 0$, meaning $0 = x_i - \lambda + \alpha_i \implies x_i - \lambda \le 0$.
    \item If $y^*_i = 1$, then $\alpha_i = 0$ and $\beta_i \ge 0$, meaning $1 = x_i - \lambda - \beta_i \implies x_i - \lambda \ge 1$.
\end{itemize}
Combining these mutually exclusive cases, $y^*_i$ is exactly the clamped shifted value given  $\lambda$:
\[
    y^*_i(\lambda) = \min(\max(x_i - \lambda, 0), 1)
\]

To satisfy the primal constraint $\sum_{i=1}^n y^*_i = k$, we must find the optimal dual variable $\lambda^*$ such that
\[
    F(\lambda; \mathbf{x}) = \sum_{i=1}^n \min(\max(x_i - \lambda, 0), 1) = k.
\]

The clamped shifted sum $F(\lambda; \mat{x})$ is a continuous, decreasing function of $\lambda$.
It is composed of a sum of piecewise linear functions $f_i(\lambda) \defeq \min(\max(x_i - \lambda, 0), 1).$
Each term $f_i(\lambda)$ changes its slope at exactly two threshold points:
$\lambda = x_i-1$ and $\lambda = x_i $.
Therefore, the entire function $F(\lambda ; \mat{x})$ is piecewise linear,
and its complete set of breakpoints where the slope of the sum changes is strictly
contained in the set
\[
    \texttt{points} = \bigcup_{i=1}^{n} \{x_i - 1, x_i\}.
\]
Between any two consecutive breakpoints, no individual term $f_i(\lambda)$ changes its individual state.
It remains identically $0$, identically $1$, or linearly decreasing with slope $-1$.
Thus, $F(\lambda ; \mat{x})$ is linear in any interval
$[\texttt{points}[i], \texttt{points}[i+1]]$.

Since $F(\lambda ; \mat{x})$ is decreasing, we isolate the two adjacent breakpoints
that bound the target value $k$ by sorting the array of $2n$ points and using binary search.
The binary search finds index $i^*$ such that
\[
    F(\texttt{points}[i] ; \mat{x}) \ge k > F(\texttt{points}[i + 1] ; \mat{x}).
\]
Note that at the endpoints we have $F(\min(\texttt{points}) ; \mat{x}) = n$ and $F(\max(\texttt{points}) ; \mat{x}) = 0$.

Let $\lambda_1 = \texttt{points}[i^*]$
and $\lambda_2 = \texttt{points}[i^* + 1]$.
Since $F(\lambda ; \mat{x})$ is linear on the interval $[\lambda_1, \lambda_{2}]$,
the exact root $\lambda^*$ can be found through linear interpolation:
\[
    \frac{\lambda^* - \lambda_1}{\lambda_2 - \lambda_1} = \frac{F(\lambda_1 ; \mat{x}) - k}{F(\lambda_1 ; \mat{x}) - F(\lambda_2 ; \mat{x})}.
\]
Solving for $\lambda^*$ gives the linear interpolation on Line~7 of the algorithm:
\[
    \lambda^*
    =
    \lambda_1 + (\lambda_2 - \lambda_1) \frac{F(\lambda_1 ; \mat{x}) - k}{F(\lambda_1 ; \mat{x}) - F(\lambda_2 ; \mat{x})}.
\]
Once $\lambda^*$ is computed, the projection step $y_{i} = \min(\max(x_i - \lambda^*, 0), 1)$ guarantees that all constraints are satisfied. This establishes the correctness of \Cref{alg:euclidean_projection_hypersimplex}.

\paragraph{Running time.}
We first sort \texttt{points} in $O(n\log n)$ time.
Then we run binary search for $O(\log n)$ steps.
The evaluation of $F(\lambda ; \mat{x})$ in each step takes $O(n)$ time,
so it takes $O(n\log n)$ time to find $i^*$.
It then takes $O(n)$ time to compute $\lambda^*$ and $\mat{y} \in \Delta_{n,k}$,
so the overall running time is $O(n \log n)$.
\end{proof}

\subsection{Proof of \Cref{lem:convergence_rate}}
\label{app:accelerated_projected_gradient_ascent}

\ConvergenceRate*

\begin{proof}
Lines 7--16 of \Cref{alg:accelerated_relax_and_round} are the \emph{fast iterative shrinkage-thresholding algorithm} (FISTA) for
\[
    \max_{\mat{x} \in \Delta_{n,k}} \SmoothC(\mat{x}).
\]
\citet[Theorem 4.4]{beck2009fast} showed that the convergence rate of this first-order accelerated gradient method is
\[
    \SmoothC(\widetilde{\mat{x}}_\mu^*) - \SmoothC(\mat{x}^{(t)}) \le \frac{2 \norm{\mat{x}^{(0)} - \widetilde{\mat{x}}_\mu^*}_{2}^2}{\eta (t + 1)^2},
\]
where $\eta \le 1/L$ is a constant step size and $L$ is a Lipschitz constant of $\nabla \SmoothC(\mat{x})$.
Using \Cref{lem:smooth_c_guarantees},
\[
    L = \frac{\degreeR}{4 \mu}.
\]
Since $\mat{x}^{(0)} = \xgreedy$ and $\widetilde{\mat{x}}_\mu^* \in \Delta_{n,k}$,
\begin{align*}
    \norm{\mat{x}^{(0)} - \widetilde{\mat{x}}_\mu^*}_{2}^2
    &=
    \norm{\xgreedy}_{2}^2 -2\xgreedy^\intercal \widetilde{\mat{x}}_\mu^* + \norm{\widetilde{\mat{x}}_\mu^*}_{2}^2 \\
    &\le
    \norm{\xgreedy}_{2}^2 + \norm{\widetilde{\mat{x}}_\mu^*}_{2}^2 \\
    &\le
    k + k \\
    &=
    2k.
\end{align*}
We set $\eta = 1/L$ in Line~8, so for any $t \ge 1$,
\[
    \SmoothC(\widetilde{\mat{x}}_\mu^*) - \SmoothC(\mat{x}^{(t)}) \le \frac{k \degreeR }{\mu (t + 1)^2}.
\]
Combining this with the approximation error in \Cref{lem:smooth_c_guarantees} gives
\begin{align*}
    C(\mat{x}^*) - \widetilde{C}_{\mu}(\mat{x}^{(t)})
    &=
    \bracks*{C(\mat{x}^*) - \widetilde{C}_{\mu}(\widetilde{\mat{x}}_{\mu}^*)}
    + \bracks*{\widetilde{C}_{\mu}(\widetilde{\mat{x}}^*_{\mu}) - \widetilde{C}_{\mu}(\mat{x}^{(t)})}, \\
    &\le
    \mu \log \degreeR + \frac{k \degreeR}{\mu (t + 1)^2},
\end{align*}
which completes the proof.
\end{proof}

\subsection{Proof of \Cref{thm:relative_error_algorithm}}
\label{app:optimizing_mu}

We start with a general relative error guarantee in terms of an underestimate factor $\gamma$ for the optimal smoothing parameter.

\begin{lemma}
\label{lem:optimizing_mu}
Let $\varepsilon > 0$ and $\mu = \gamma \cdot \frac{\varepsilon C(\mat{x}^*)}{2 \log \degreeR}$, for $0 < \gamma \le 1$.
Then, for any $t \ge \frac{2}{\varepsilon C(\mat{x}^*)} \sqrt{\frac{k \degreeR \log \degreeR}{\gamma(2 - \gamma)}}$,
\[
    \widetilde{C}_{\mu}(\mat{x}^{(t)}) \ge \parens*{1 - \varepsilon} \cdot C(\mat{x}^*).
\]
\end{lemma}

\begin{proof}
To prove that $\SmoothC(\mat{x}^{(t)}) \ge (1 - \varepsilon) \cdot C(\mat{x}^*)$,
it is sufficient by \Cref{lem:convergence_rate} to have
\[
    \frac{k \degreeR}{\mu (t + 1)^2} + \mu \log \degreeR
    \le
    \varepsilon C(\mat{x}^*).
\]
Rearranging the inequality to isolate $(t+1)^2$,
\begin{align*}
    \frac{k \degreeR}{\mu (t + 1)^2} + \mu \log \degreeR \le \varepsilon C(\mat{x}^*)
    &\iff
    \frac{k \degreeR}{\mu (t + 1)^2} \le \varepsilon C(\mat{x}^*) - \mu \log \degreeR \\
    &\iff
    \frac{k \degreeR}{\mu (\varepsilon C(\mat{x}^*) - \mu \log \degreeR) } \le (t + 1)^2.
\end{align*}

To minimize the required number of steps, we compute the optimal $\mu^*$
that maximizes the denominator of the left-hand side:
\begin{equation}
\label{eqn:mu_star}
    \mu^* = \frac{\varepsilon C(\mat{x}^*)}{2 \log \degreeR}.
\end{equation}
However, the true value of $C(\mat{x}^*)$ is unknown, so we use an underestimate
$\mu = \gamma \mu^*$, where $0 < \gamma \le 1$ is a multiplier representing the error in our guess.
Substituting back into the inequality,
\[
    \frac{k \degreeR}{\gamma \mu^* (2 \mu^* \log \degreeR - \gamma \mu^* \log \degreeR)}
    =
    \frac{k \degreeR}{{\mu^*}^2 \log \degreeR \cdot \gamma (2 - \gamma)}
    \le (t + 1)^2.
\]
Observe that $\gamma (2 - \gamma)$ is maximized at $\gamma = 1$, but any $\gamma \in (0, 2)$ gives a valid bound.
It follows that
\begin{align*}
\label{eqn:T_lower_bound}
    \sqrt{\frac{k \degreeR}{{\mu^*}^2 \log \degreeR \cdot \gamma (2 - \gamma)}}
    &=
    \frac{2 \log \degreeR}{\varepsilon C(\mat{x}^*)} \sqrt{\frac{k \degreeR}{\log \degreeR \cdot \gamma (2 - \gamma)}} \\
    &=
    \frac{2}{\varepsilon C(\mat{x}^*)} \sqrt{\frac{k \degreeR \log \degreeR}{\gamma(2 - \gamma)}}
    \le
    t + 1.
\end{align*}
The final inequality holds because of our assumption on $t$ in the hypothesis.
\end{proof}

\begin{lemma}[Lower bounding $C(\mat{x}^*)$]
\label{lem:bounding_fractional_opt}
For any input, we have
\begin{align*}
    C(\mat{x}^*) \ge 0.43 \cdot \frac{\sqrt{k \degreeR}}{n}.
\end{align*}
\end{lemma}

\begin{proof}
We start with a simple lower bound
\begin{align*}
    C(\mat{x}^*)
    &\ge
    C\parens*{\frac{k}{n} \cdot \mat{1}_n} \\
    &=
    \sum_{j \in R} w_{j} \varphi\parens*{\frac{k}{n} \cdot \deg(j)} \\
    &\ge \frac{k}{n} \sum_{j \in R} w_j \\
    &= \frac{k}{n}.
\end{align*}

To improve this bound, partition the right nodes based on whether their degree exceeds $n/k$.
Define the weighted fraction of \emph{high-degree nodes} as
\[
    \alpha \defeq \sum_{j \in R, \deg(j) > n/k} w_j.
\]
Define the effective degree of the \emph{low-degree nodes} as
\[
    \degreeR^{\textnormal{low}} \defeq \sum_{j \in R, \deg(j) \le n/k} w_j \cdot \deg(j).
\]
By allocating each left node $x_i \gets k / n$, we can lower bound the optimal value by
\begin{align*}
    C(\mat{x}^*) &\ge \alpha + \sum_{j \in R, \deg(j) \le n/k} w_{j} \varphi\parens*{\frac{k}{n} \cdot \deg(j)} \\ 
    &= \alpha + \sum_{j \in R, \deg(j) \le n/k} w_{j} \cdot \frac{k}{n} \cdot \deg(j) \\
    &= \alpha + \frac{k}{n} \cdot \degreeR^{\textnormal{low}}.
\end{align*}

Since both $k$ and $\degreeR$ are at most $n$, we have the lower bound
\begin{equation*}
    C(\mat{x}^*) \ge \alpha \cdot \frac{n}{n} \ge \alpha \cdot \frac{\sqrt{k \degreeR}}{n}. 
\end{equation*}

To prove the claim, we give a second lower bound in terms of $\alpha$ and take the maximum of the two lower bounds.

Observe that we can upper bound $\degreeR$ by
\[
    \degreeR \le \alpha n + \degreeR^{\textnormal{low}}.
\]
Now we have all the ingredients to prove the second bound:
\begin{align*}
    \frac{\sqrt{k \degreeR}}{C(\mat{x}^*)}
    &\le \frac{\sqrt{k}}{C(\mat{x}^*)} \parens*{\sqrt{\alpha n + \degreeR^{\textnormal{low}}}} \\
    &\le \sqrt{k} \cdot \parens*{\frac{\sqrt{\alpha n}}{C(\mat{x}^*)} + \frac{\sqrt{\degreeR^{\textnormal{low}}}}{C(\mat{x}^*)}} \\
    &\le \sqrt{k} \cdot \parens*{\frac{\sqrt{\alpha n}}{\max\{\alpha, k/n\}} + \frac{\sqrt{\degreeR^{\textnormal{low}}}}{\degreeR^{\textnormal{low}} \cdot k/n}}     \\
    &= \min \set*{\sqrt{\frac{kn}{\alpha}}, n \sqrt{\frac{\alpha n}{k}} } + \frac{n}{\sqrt{k \cdot \degreeR^{\textnormal{low}}}} \\
    &\le \sqrt{ \sqrt{\frac{kn}{\alpha}} \cdot n \sqrt{\frac{\alpha n}{k}} } + \frac{n}{\sqrt{k \cdot \degreeR^{\textnormal{low}}}} \\
    &= n + \frac{n}{\sqrt{k \cdot \degreeR^{\textnormal{low}}}} \\
    &\le n + \frac{n}{\sqrt{1-\alpha}},
\end{align*}

where the last inequality holds because $k \ge 1$ and each right node $j \in R$ has $\deg(j) \ge 1$, hence
\[
    \degreeR^{\textnormal{low}} \ge 1 - \alpha
    \implies
    k\cdot \degreeR^{\textnormal{low}} \ge 1 - \alpha.
\]

Putting the two bounds together, we have
\[
    C(\mat{x}^*) \ge \max\set*{\alpha, \frac{1}{1 + \frac{1}{\sqrt{1-\alpha}}}} \cdot \frac{\sqrt{k \degreeR}}{n}
\]
For $\alpha \ge 0$, the term $\tau = \max\set*{\alpha, \frac{1}{1 + \frac{1}{\sqrt{1-\alpha}}}}$ is minimized around $\alpha \approx 0.4302$.
To prove the bound, observe that if $\alpha < 0.43$, then the second term in $\tau$ is at least $1/(1 + \frac{1}{\sqrt{1-0.43}}) > 0.43$.
If $\alpha \ge 0.43$, then we clearly have $\tau \ge 0.43$.
Therefore, we have the final lower bound
\[
    C(\mat{x}^*) \ge 0.43 \cdot \frac{\sqrt{k \degreeR}}{n},
\]
as desired.
\end{proof}

\OptimizingMu*

\begin{proof}
Building on \Cref{lem:optimizing_mu}, we use $\xgreedy$ in our underestimate of $\mu^*$:
\[
    \mu \gets \frac{\varepsilon C(\xgreedy)}{2 \log \degreeR}.
\]
It follows from our definition $\mu = \gamma \mu^*$ and \Cref{lem:xgreedy} that
\[
    \gamma = \frac{C(\xgreedy)}{C(\mat{x}^*)} \ge 1 - 1/e.
\]
Using the upper bound $\frac{\sqrt{k \degreeR}}{C(\mat{x}^*)} \le \frac{n}{0.43}$ in \Cref{lem:bounding_fractional_opt},
we bound the number of required iterations by
\begin{align*}
    \frac{2}{\varepsilon C(\mat{x}^*)} \sqrt{\frac{k \degreeR \log \degreeR}{\gamma(2 - \gamma)}}
    &\le
    \frac{2}{\varepsilon C(\mat{x}^*)} \sqrt{\frac{k \degreeR \log \degreeR}{(1 - \sfrac{1}{e})(1 + \sfrac{1}{e})}} \\
    &\le
    \frac{2 n}{0.43 \varepsilon} \sqrt{\frac{\log \degreeR}{(1 - \sfrac{1}{e})(1 + \sfrac{1}{e})}} \\
    &\le
    \frac{2 n}{0.43\varepsilon} \sqrt{\frac{\log n}{(1 - \sfrac{1}{e})(1 + \sfrac{1}{e}) }} \\
    &\le \frac{5.1 n \sqrt{\log n}}{\varepsilon}. \qedhere
\end{align*}
\end{proof}

%% file: app_convex_combination.tex
\subsection{Rounding in the hypersimplex}
\label{app:rounding}

\subsubsection{Carathéodory decomposition}
\label{app:convex_combination}

A key insight for our rounding strategy is Carathéodory's theorem, which states that any point~$\mat{x}$ in an $n$-dimensional convex hull $\textnormal{Conv}(P) \subseteq \R^n$ can be written as a convex combination of at most $n+1$ vertices of $P$.
In our case, we want to decompose a fractional solution $\mat{x}^{(T)} \in \Delta_{n,k}$,
where $\Delta_{n,k}$ is the hypersimplex (the convex hull of all binary vectors in $\{0, 1\}^n$ with exactly $k$ ones).
Even though the hypersimplex lives in $\R^n$, the constraint $\sum_{i=1}^n x_i = k$
restricts it to an $(n-1)$-dimensional affine subspace.
Thus, the effective dimension is $n-1$ and Carathéodory's theorem
guarantees that any point $\mat{x} \in \Delta_{n,k}$ can be written as a convex combination of at most $n$ vertices.

Computing a Carathéodory decomposition can be expensive, often requiring linear programming,
which is often too slow for large-scale coverage problems.
To solve this, we use the recent combinatorial breakthrough in \citet[Algorithm 1]{karalias2025geometric},
as it is a highly efficient way to decompose $\mat{x} \in \Delta_{n,k}$ into at most $n$ vertices.

\begin{theorem}[{\citet[Theorem 4.3]{karalias2025geometric}}]
\label{thm:decomposition}
For $\mat{x} \in \Delta_{n,k}$, \Cref{alg:caratheodory-representation} terminates in  $r\leq n$ iterations and returns probability-vector pairs $\{(\alpha_t, \mat{v}_t)\}_{t=1}^{r}$ such that $\mat{x} = \sum_{t=1}^{r} \alpha_t \mat{v}_t$, where $\mat{v}_t$ is a binary vector with exactly $k$ ones.
This algorithm can be implemented to run in $O(n k \log n)$ time.
\end{theorem}

\begin{algorithm}
\caption{Convex combination representation of $\mat{x} \in \Delta_{n,k}$.}
\label{alg:caratheodory-representation}
\begin{algorithmic}[1] 
\Procedure{\texttt{ConvexCombination}}{vector $\mat{x} \in \Delta_{n,k}$, integer $k$}

    \State $\texttt{decomposition} \gets \varnothing$
    \State $w_1 \gets 1$  \hfill \Comment{Remaining probability mass}
    \State $\mat{x}^{(1)} \gets \mat{x}$
    \For{ $t = 1, 2, \dots$}
        \State Find indices $I_t$ of the $k$ largest components of $\mat{x}^{(t)}$ \label{extract}
        \State $\mat{v}^{(t)} \gets \mat{1}_{I_{t}}$
        \State Compute the max weight we can give to $\mat{v}^{(t)}$ without violating rescaling constraints:
        \label{max_beta}
        \[
            \beta_t \gets \min\parens*{\min_{i \in I_t} x^{(t)}_i, \min_{j \not\in I_t} (1 - x^{(t)}_j)}
        \]
        \State Add $(\beta_t w_t, \mat{v}^{(t)})$ to \texttt{decomposition}
        \If{$\beta_t = 1$}
            \State \textbf{break}
        \EndIf
        \State Compute residual vector
        $
            \mat{x}^{(t+1)} \gets \frac{\mat{x}^{(t)} - \beta_t \mat{v}^{(t)}}{1 - \beta_t}
        $ \label{rescale}
        \State $w_{t+1} \gets w_t (1 - \beta_t)$
    \EndFor
    \State \Return \texttt{decomposition}
\EndProcedure
\end{algorithmic}
\end{algorithm}

\paragraph{Main idea.} 
The algorithm decomposes $\mat{x}$ by iteratively peeling off a vertex $\mat{v}^{(t)}$ supported on the~$k$ largest elements of $\mat{x}$ (Line~\ref{extract}), assigning it the maximum possible weight, and solving the problem for the ``residual'' vector $\mat{x}^{(t+1)}$.
To apply the decomposition step recursively, the algorithm maintains a key loop invariant: the residual vector $\mat{x}^{(t)}$ is rescaled so that its $L_1$ norm remains exactly $k$ (Line~\ref{rescale}).
The maximum possible weight $\beta_t$ is specifically chosen to guarantee that at least one coordinate of the residual vector becomes exactly $0$ or $1$ (Line~\ref{max_beta}). Consequently, the scaling factor $\frac{1}{1-\beta_t}$ ensures that the elements of the residual vector stay safely within the interval $[0, 1]$ (Line~\ref{rescale}). 
 Since at least one index of $\mat{x}$ is eliminated or fixed at each iteration, the algorithm runs for at most $n$ iterations.
 Finally, the probability allocated to the vector $\mat{v}^{(t)}$ at iteration $t$ is $\alpha_t =  \beta_t \omega_t =  \beta_t \prod_{t' = 1}^{t-1} (1-\beta_{t'})$.

\input{app_swap_round}

%% file: app_swap_round.tex
\subsubsection{Randomized swap rounding}
\label{app:swap_round}

Now we give a tailored version of the \SwapRound algorithm  
\citep{chekuri2010dependent} for points $\mat{x} \in \Delta_{n, k}$.

\begin{algorithm}[H]
\caption{Swap rounding for points in the hypersimplex $\mat{x} \in \Delta_{n,k}$.}
\label{alg:swap_round}
\begin{algorithmic}[1] 

\Procedure{\texttt{MergeBases}}{$(\alpha_1, \mat{v}^{(1)}), (\alpha_2, \mat{v}^{(2)})$}
    \State $\gamma \gets \alpha_1 + \alpha_2$
    \State $p \gets \alpha_1 / \gamma$
    
    \While{$\mat{v}^{(1)} \neq \mat{v}^{(2)}$}
        \State Pick an element $i \in \mat{v}^{(1)} \setminus \mat{v}^{(2)}$
        \State Pick an element $j \in \mat{v}^{(2)} \setminus \mat{v}^{(1)}$ 
        \State With probability $p$: \Comment{Perform a randomized swap}
        \State \indent $\mat{v}^{(2)} \gets \mat{v}^{(2)} + \{i\} - \{j\}$ 
        \State Otherwise (with probability $1-p$):
        \State \indent $\mat{v}^{(1)} \gets \mat{v}^{(1)} - \{i\} + \{j\}$ 
    \EndWhile
    
    \State \Return $(\gamma, \mat{v^{(1)}})$ 
\EndProcedure

\Procedure{\SwapRound}{vector $\mat{x}\in \Delta_{n,k}$}
\State $k \gets \sum_{i=1}^n x_i$
\State $\{(\alpha_i, \mat{v}^{(i)})\}_{i=1}^r \gets \texttt{ConvexCombination}(\mat{x}, k)$ \label{decomposition_step} \Comment{$\mat{x} = \sum_{i=1}^r \alpha_i \mat{v}^{(i)}$}
\State Initialize $\mat{v} \gets \mat{v}^{(1)}, \alpha \gets \alpha_1$
\For{$i = 2, 3, \dots, r$}
    \State $(\alpha, \mat{v}) \gets \texttt{MergeBases}{((\alpha, \mat{v}), (\alpha_i, \mat{v}^{(i)}))}$
\EndFor
\State Let $S \gets \{i \in [n] : v_i = 1\}$
\State \Return $S$

\EndProcedure
\end{algorithmic}
\end{algorithm}

\HypersimplexSwapRound*

\begin{proof}
The first step of \SwapRound calls the vertex-peeling algorithm of
\citet{karalias2025geometric} to get the
convex combination $\mat{x} = \sum_{i=1}^r \alpha_i \mat{v}^{(i)}$,
where $r \le n$ by \Cref{thm:decomposition}.
Since $C(S)$ is submodular and $F(\mat{x})$ is its multilinear extension,
applying \cite[Theorem 2.1]{chekuri2010dependent} gives
\[
    \E[C(\SwapRound(\mat{x}))] \geq F(\mat{x}).
\]
In each iteration of \MergeBases, the swap operation for $\mat{v}^{(1)}$ and $\mat{v}^{(2)}$ removes exactly one element (changing from 1 to 0) and adds exactly one element (changing from 0 to 1).
Therefore, the total number of ones remains constant at $k$ throughout every iteration of the while loop
and the output of \MergeBases is a valid vector with $k$ ones and $n-k$ zeros. 

\paragraph{Running time.}
Finding the Carathéodory decomposition takes $O(nk\log{n})$ time by \Cref{thm:decomposition}.
There are at most $n-1$ calls to \MergeBases,
and in each call to \MergeBases the Hamming distance between the input vectors is at most $2k$.
There is a simple perfect matching between the elements in the
differences $\mat{v}^{(1)} \setminus \mat{v}^{(2)}$
and $\mat{v}^{(2)} \setminus \mat{v}^{(1)}$
that we can use for deciding which $i$ and $j$ to swap in each step,
which we compute in $O(k \log k)$ time using balanced binary search trees.
Therefore, \MergeBases takes $O(k\log{k})$ time,
which means \SwapRound runs in $O(n k \log n)$ time.
\end{proof}

%% file: app_rounding_ratio.tex
\section{Analysis for \Cref{sec:approximation_ratios}}
\label{app:approximation_ratios}

In this section, we evaluate the Poisson concavity ratio $\alpha_\varphi$ for different families of reward functions,
i.e., the approximation ratio of the relax-and-round algorithms:
\[
  \alpha_{\varphi}
  = 
  \inf_{x \in \Z_{\ge 1}}
  \frac{\E[\varphi(\Pois(x))]}{\varphi(x)}.
\]

\subsection{Properties of Poisson random variables}

\begin{lemma}
\label{lem:expected_min}
If $X \sim \Pois(x)$ and $c \in \Z_{\ge 0}$, then
$
    \E[\min(X, c)] = x \Pr(X \le c - 2) + c \Pr(X \ge c).
$
\end{lemma}

\begin{proof}
Observe that
\begin{align*}
    \E[\min(X, c)]
    &=
    \sum_{k=0}^\infty \min(k, c) \frac{x^k e^{-x}}{k!} \\ 
    &=
    \sum_{k=1}^{c - 1} k \frac{x^k e^{-x}}{k!} + \sum_{k=c}^\infty c \frac{x^k e^{-x}}{k!} \\
    &=
    x \sum_{k=1}^{c - 1} \frac{x^{k-1} e^{-x}}{(k-1)!} + c \Pr(X \ge c) \\
    &=
    x \sum_{k=0}^{c - 2} \frac{x^{k} e^{-x}}{k!} + c \Pr(X \ge c) \\
    &=
    x \Pr(X \le c - 2) + c \Pr(X \ge c). \qedhere
\end{align*}
\end{proof}

\begin{lemma}
\label{lem:poisson_reward_derivative}
Let $X \sim \Pois(x)$.
For any nonnegative concave function $\varphi : \Z_{\ge 0} \rightarrow \R$,
\[
    \frac{d}{dx} \E[\varphi(X)] = \E[\varphi(X + 1) - \varphi(X)].
\]
\end{lemma}

\begin{proof}
Observe that
\begin{align*}
    \frac{d}{dx} \E[\varphi(X)]
    &=
    \frac{d}{dx} \sum_{k=0}^\infty \varphi(k) \Pr(X = k) \\
    &=
    \sum_{k=0}^\infty \varphi(k) \frac{d}{dx}  \parens*{\frac{e^{-x} x^k}{k!}} \\
    &=
    \sum_{k=0}^\infty \varphi(k) \parens*{\frac{k e^{-x} x^{k-1}}{k!} - \frac{e^{-x} x^k}{k!}} \\
    &=
    \sum_{k=0}^\infty \varphi(k) \frac{e^{-x} x^{k-1}}{(k-1)!} - \sum_{k=0}^\infty \varphi(k) \frac{e^{-x} x^k}{k!} \\
    &=
    \sum_{k=0}^\infty \varphi(k + 1) \frac{e^{-x} x^{k}}{k!} - \sum_{k=0}^\infty \varphi(k) \frac{e^{-x} x^k}{k!} \\
    &=
    \E[\varphi(X + 1)] - \E[\varphi(X)],
\end{align*}
which completes the proof.
\end{proof}

\begin{corollary}
\label{cor:expected_min_derivative}
If $X \sim \Pois(x)$, then
\[
    \frac{d}{dx} \E[\min(X, c)] = \Pr(X \le c - 1).
\]
\end{corollary}

\begin{proof}
Applying \Cref{lem:poisson_reward_derivative},
\begin{align*}
    \frac{d}{dx} \E[\min(X, c)]
    &=
    \E[\min(X + 1, c) - \min(X, c)] \\
    &=
    \E[\mat{1}_{X \le c - 1}] \\
    &=
    \Pr(X \le c - 1),
\end{align*}
as desired.
\end{proof}



\subsection{Proof of \Cref{thm:log_ratio}}
\label{app:log_ratio}

\LogRatio*

\begin{proof}
We use the formula for $\alpha_{\varphi}$ in \eqref{eqn:poisson_concavity_ratio}
to show that
$\inf_{x \in \Z_{\ge 0}} \alpha_{\varphi}(x) = \alpha_{\varphi}(1)$.
First, we numerically verify that
\begin{align*}
    \alpha_{\varphi}(1) &\approx 0.8272 \\
    \alpha_{\varphi}(2) &\approx 0.8902 \\
    \alpha_{\varphi}(3) &\approx 0.9240 \\
    \alpha_{\varphi}(4) &\approx 0.9440
\end{align*}
We proceed by showing that $\alpha_{\varphi}(x) \ge 0.83$ for $x \ge 4$.

Let $X \sim \Pois(x)$.
Using the inequality $\log(z) \ge 1 - 1/z$, which holds for all $z > 0$, we have
\begin{align*}
    \log\parens*{\frac{1 + X}{1 + x}} \ge 1 - \frac{1+x}{1+X}.
\end{align*}
Applying logarithm rules,
\begin{align*}
    \log(1 + X) \ge \log(1 + x) + 1 - \frac{1 + x}{1 + X}.
\end{align*}
By the linearity of expectation,
\begin{align}
\label{eqn:log_useful_lower_bound}
    \E\bracks*{\log(1 + X)} \ge \log(1 + x) + 1 - (1 + x) \E\bracks*{\frac{1}{1 + X}}.
\end{align}

The expected value on the right-hand side of \eqref{eqn:log_useful_lower_bound}
has an exact closed-form solution:
\begin{align*}
    \E\bracks*{\frac{1}{1 + X}}
    &=
    \sum_{k = 0}^\infty \frac{1}{1 + k} \parens*{\frac{x^k e^{-x}}{k!}} \\
    &=
    \frac{1}{x} \sum_{k = 0}^\infty \parens*{\frac{x^{k+1} e^{-x}}{(k+1)!}} \\
    &=
    \frac{1}{x} \parens*{1 - \Pr(X = 0)} \\
    &=
    \frac{1 - e^{-x}}{x}.
\end{align*}
Putting everything together,
\begin{align*}
    \alpha_{\varphi}(x)
    &=
    \frac{\E\bracks*{\log(1 + X)}}{\log(1 + x)} \\
    &\ge
    1 + \frac{1}{\log(1 + x)} - \frac{1 + x}{x} \cdot \frac{1 - e^{-x}}{\log(1 + x)} \\
    &\ge
    1 + \frac{1}{\log(1 + x)} - \frac{1 + x}{x} \cdot \frac{1}{\log(1 + x)} \\
    &=
    1 - \frac{1}{x \log(1 + x)}.
\end{align*}
The function $f(x) = 1 - \frac{1}{x \log(1 + x)}$
is increasing for $x > 0$.
Therefore, if $x \ge 4$, then
\[
    \alpha_{\varphi}(x)
    \ge
    1 - \frac{1}{x \log(1 + x)}
    \ge
    f(4)
    \approx
    0.8446
    \ge 0.83,
\]
which proves the claim.
\end{proof}

\subsection{Proof of \Cref{thm:piecewise_linear_ratio}}

\subsubsection{Warm-up: Maximum coverage problem}
Let $X \sim \Pois(x)$ and
consider the coverage function $\varphi(x) = \min(x, 1)$.
Then,
\begin{align*}
    \E[\varphi(X)]
    &=
    \sum_{k=0}^\infty \varphi(k) \frac{x^k e^{-x}}{k!} \\
    &=
    e^{-x} \parens*{0 + \frac{x}{1!} + \frac{x^2}{2!} + \frac{x^3}{3!} + \dots} \\
    &=
    e^{-x} \parens*{e^{x} - 1} \\
    &=
    1 - e^{-x}.
\end{align*}
It follows that
\begin{align*}
    \alpha_{\varphi}
    &=
    \inf_{x \in \Z_{\ge 1}} \frac{\E[\varphi(X)]}{\varphi(x)}
    =
    \inf_{x \in \Z_{\ge 1}} \frac{1 - e^{-x}}{1}
    =
    1 - \frac{1}{e}
    \approx
    0.632.
\end{align*}

\subsubsection{Maximum multi-coverage}
\label{subsec:c_coverage}

For any positive integer $c$,
define the \emph{$c$-multi-coverage} function $\varphi(x) = \min(x, c)$.
By \Cref{lem:expected_min}, our goal is to compute
\begin{align}
\label{eqn:alpha_c_coverage}
    \alpha_{\varphi}
    &=
    \inf_{x \in \Z_{\ge 1}} \frac{\E[\varphi(X)]}{\varphi(x)}
    =
    \inf_{x \in \Z_{\ge 1}} \frac{x \Pr(X \le c - 2) + c \Pr(X \ge c)}{\min(x, c)}.
\end{align}

\citet[Theorem 1.1]{barman2022tight} prove the Poisson concavity ratio for $c$-multi-coverage functions,
but we give a different proof and use some of the new inequalities for our proof of \Cref{thm:piecewise_linear_ratio}.

\begin{lemma}[{\citet[Theorem 1.1]{barman2022tight}}]
\label{lem:c_coverage_approximation_ratio}
Let $c$ be a positive integer.
If $\varphi(x) = \min(x, c)$, then
\[
    \alpha_{\varphi}
    =
    1 - \frac{c^c e^{-c}}{c!}.
\]
\end{lemma}

\begin{proof}
We start by evaluating $\alpha_{\varphi}(x)$ at $x = c$:
\begin{align*}
    \frac{c \Pr(X \le c - 2) + c \Pr(X \ge c)}{c}
    &=
    1 - \Pr(X = c - 1)
    =
    1 - \frac{c^{c-1} e^{-c}}{(c-1)!}
    =
    1 - \frac{c^{c} e^{-c}}{c!}.
\end{align*}
Next, we analyze two cases to show that this is the minimum value of $\alpha_{\varphi}(x)$ for $x > 0$.

\begin{figure}[t]
    \centering
    \includegraphics[width=0.5\textwidth]{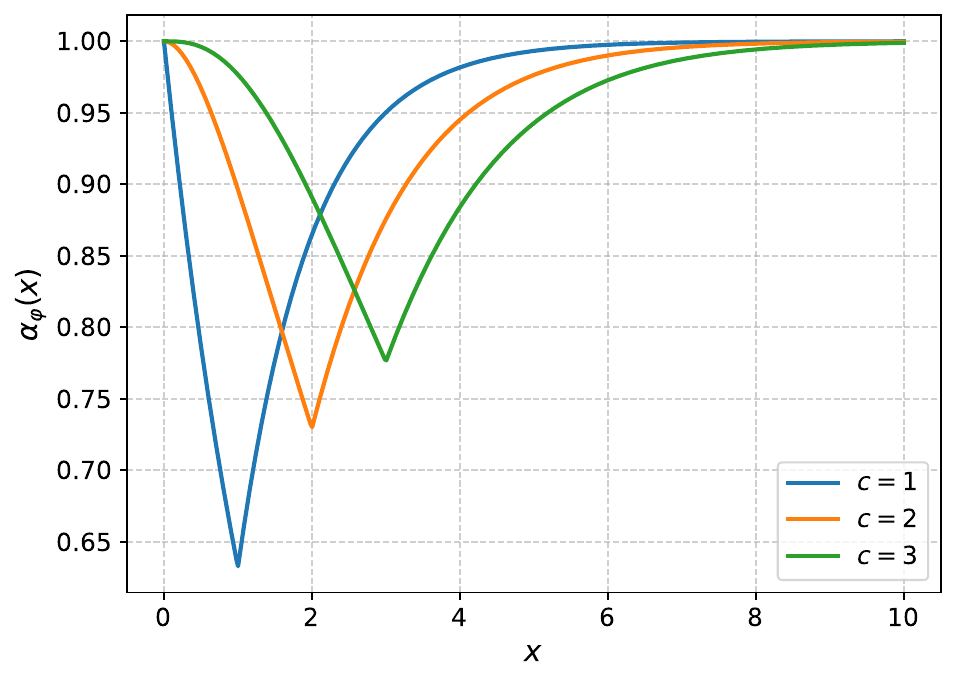}
    \caption{Plots of the Poisson concavity ratio $\alpha_{\varphi}(x) = \E_{X \sim \Pois(x)}[\varphi(X)] / \varphi(x)$ for $c$-multi-coverage functions $\varphi(x) = \min(x, c)$.}
    \label{fig:c_coverage_plots}
\end{figure}

\paragraph{Case 1 ($x \le c$):}
Let
\begin{equation}
\label{eqn:multi_coverage_f_def}
    f(x) = \frac{\E[\min(X, c)]}{x}.
\end{equation}
We prove that $f(x)$ is decreasing for $x \in (0, c]$.
Applying the quotient rule and \Cref{cor:expected_min_derivative},
\begin{align*}
    f'(x)
    &=
    \frac{x \frac{d}{dx}\E[\min(X, c)] - \E[\min(X, c)]}{x^2}
    < 0
    \iff
    x \Pr(X \le c - 1) < \E[\min(X, c)].
\end{align*}
Using the identity in \eqref{eqn:alpha_c_coverage}, it is equivalent to show that
\begin{align*}
    x \Pr(X \le c - 1) < \E[\min(X, c)]
    &\iff
    x \Pr(X \le c - 1) < x \Pr(X \le c - 2) + c \Pr(X \ge c) \\
    &\iff
    x \Pr(X = c - 1) < c \Pr(X \ge c).
\end{align*}
Observe that
\[
    x \Pr(X = c - 1)
    =
    x \cdot \frac{x^{c-1} e^{-x}}{(c-1)!}
    =
    c \cdot \frac{x^c e^{-x}}{c!}
    =
    c \Pr(X = c).
\]
It follows that
\begin{align*}
    x \Pr(X = c - 1)
    =
    c \Pr(X = c)
    <
    c \Pr(X \ge c),
\end{align*}
as desired.

\paragraph{Case 2 $(x \ge c)$:}
We prove that $\alpha_{\varphi}(x)$ is increasing for $x \ge c$.
Let
\begin{align*}
    g(x)
    &=
    \frac{\E[\min(X, c)]}{c}.
\end{align*}
Applying \Cref{cor:expected_min_derivative},
\begin{align*}
    g'(x)
    =
    \frac{1}{c} \parens*{\frac{d}{dx} \E[\min(X, c)]}
    =
    \frac{1}{c} \cdot \Pr(X \le c - 1).
\end{align*}
Therefore, $g'(x) > 0$ for all $x > 0$,
which means $g(x)$ is increasing and that $g(c)$ is the minimum value on the interval $x \ge c$.
This completes the case analysis and hence the proof.
\end{proof}

\begin{remark}
We give the values of $\alpha_{\varphi}$ for $c$-multi-coverage functions below,
which improve as $c$ increases since the problem becomes more continuous and computationally easier.
\begin{table}[H]
\caption{Poisson concavity ratios for $\varphi(x) = \min(x, c)$.}
\label{table:approximation_ratios_for_c_multi_coverage}
\centering
\vspace{-0.1cm}
\begin{tabular}{cc}
    \toprule
    $c$ & $\alpha_{\varphi}$ \\
    \midrule
    1 & 0.6321 \\
    2 & 0.7293 \\
    3 & 0.7759 \\
    4 & 0.8046 \\
    5 & 0.8245 \\
    6 & 0.8393 \\
    7 & 0.8509 \\
    8 & 0.8604 \\
    9 & 0.8682 \\
    10 & 0.8748 \\
    \bottomrule
\end{tabular}
\end{table}
\end{remark}

\vspace{-0.6cm}
\subsubsection{Piecewise linear rewards}
For any $0 \le \beta \le 1$, consider the piecewise linear reward function
\[
    \varphi(x; c, \beta)
    =
    \begin{cases}
        x & \text{if $x \le c$}, \\
        c + \beta (x - c) & \text{if $x > c$}.
    \end{cases}
\]
We can write $\varphi(x; c, \beta)$ as the convex combination
\begin{align*}
    \varphi(x; c, \beta)
    &=
    \beta x + (1 - \beta) \min(x, c).
\end{align*}

Let $X \sim \Pois(x)$.
It follows from our analysis of $c$-multi-coverage functions in \Cref{subsec:c_coverage} that
\vspace{-0.2cm}
\begin{align*}
    \E[\varphi(X)]
    &=
    \E[\beta X + (1 - \beta) \min(X, c)] \\
    &=
    \beta \E[X] + (1 - \beta) \E[\min(X, c)] \\
    &=
    \beta x + (1 - \beta) \bracks*{x \Pr(X \le c - 2) + c \Pr(X \ge c)}.
\end{align*}

To compute the Poisson concavity ratio, we need to analyze
\begin{align}
\label{eqn:piecewise_linear_alpha_x}
    \alpha_{\varphi}(x; c, \beta)
    =
    \frac{\beta x + (1 - \beta) \bracks*{x \Pr(X \le c - 2) + c \Pr(X \ge c)}}{\beta x + (1 - \beta) \min(x, c)}.
\end{align}

\PiecewiseLinear*

\begin{figure}[t]
    \centering
    \includegraphics[width=0.5\textwidth]{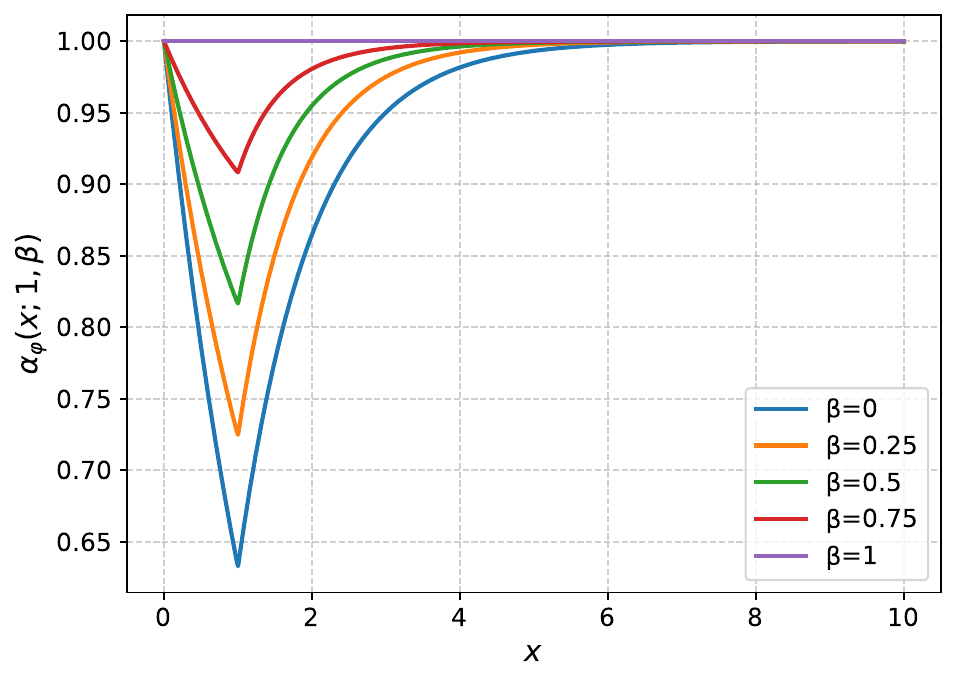}
    \caption{Plots of $\alpha_{\varphi}(x; c, \beta) = \E_{X \sim \Pois(x)}[\varphi(X; c, \beta)] / \varphi(x; c, \beta)$
    for $c = 1$ and different $\beta \in [0, 1]$.}
\end{figure}

\begin{proof}
Our proof is analogous to the proof of \Cref{lem:c_coverage_approximation_ratio}.
First, we evaluate $\alpha_{\varphi}(x; c, \beta)$ at $x = c$:
\begin{align*}
    \alpha_{\varphi}(c; c, \beta)
    &=
    \beta + (1 - \beta) \bracks*{\Pr(X \le c - 2) + \Pr(X \ge c)} \\
    &=
    \beta + (1 - \beta) \bracks*{1 - \Pr(X = c - 1)} \\
    &=
    \beta + (1 - \beta) \bracks*{1 - \frac{c^{c-1} e^{-c}}{(c-1)!}} \\
    &=
    1 - (1 - \beta) \frac{c^{c} e^{-c}}{c!}.
\end{align*}
We proceed by cases to show that the minimum value of $\alpha_{\varphi}(x; c, \beta)$ occurs at $x = c$.

\paragraph{Case 1 $(x \le c)$:}
If $x \le c$, then
\begin{align*}
    \alpha_{\varphi}(x; c, \beta)
    &=
    \frac{\beta x + (1 - \beta) \bracks*{x \Pr(X \le c - 2) + c \Pr(X \ge c)}}{x} \\
    &=
    \beta + (1 - \beta) f(x),
\end{align*}
where $f(x)$ is defined in \eqref{eqn:multi_coverage_f_def}.
We showed that $f(x)$ decreases for $x \le c$
in the proof of \Cref{lem:c_coverage_approximation_ratio},
so then $\alpha_{\varphi}(x; c, \beta)$ is decreasing for $x \in (0, c]$.

\paragraph{Case 2 $(x \ge c)$:}
If $x \ge c$, we need to analyze
\begin{align*}
    h(x)
    &=
    \frac{\beta x + (1 - \beta) \E[\min(X, c)]}{ \beta x + (1 - \beta) c}.
\end{align*}
Using the quotient rule and \Cref{cor:expected_min_derivative},
\begin{align*}
    h'(x)
    &=
    \frac{\bracks*{\beta + (1 - \beta) \Pr(X \le c - 1)} \cdot (\beta x + (1 - \beta) c) - \bracks*{\beta x + (1 - \beta) \E[\min(X, c)]} \cdot \beta}{(\beta x + (1 - \beta) c)^2}.
\end{align*}
Denoting the numerator by $N(x)$, we have
\begin{align*}
    N(x)
    &=
    \bracks*{\beta + (1 - \beta) \Pr(X \le c - 1)} \cdot (\beta x + (1 - \beta) c) - \bracks*{\beta x + (1 - \beta) \E[\min(X, c)]} \cdot \beta \\
    &=
    \beta^2 x + \beta(1-\beta)c + (1 - \beta) \Pr(X \le c - 1) (\beta x + (1 - \beta) c) - \beta^2 x - \beta (1 - \beta) \E[\min(X, c)] \\
    &=
    (1 - \beta) \cdot \bracks*{\beta c + \Pr(X \le c - 1) (\beta x + (1 - \beta) c) - \beta \E[\min(X, c)]} \\
    &=
    (1 - \beta) \cdot \bracks*{\beta c + \Pr(X \le c - 1) (\beta x + (1 - \beta) c) - \beta (x \Pr(X \le c-2) + c \Pr(X \ge c))} \\
    &=
    (1 - \beta) \cdot \bracks*{\beta c + \Pr(X \le c - 1) (\beta x + (1 - \beta) c) - \beta (x \Pr(X \le c-2) + c (1 - \Pr(X \le c - 1)))} \\
    &=
    (1 - \beta) \cdot \bracks*{\Pr(X \le c - 1) (\beta x + (1 - \beta) c) - \beta (x \Pr(X \le c-2) - c\Pr(X \le c - 1))} \\
    &=
    (1 - \beta) \cdot \bracks*{\Pr(X \le c - 1) (\beta x + (1 - \beta) c) - \beta (x \Pr(X \le c-1) - x \Pr(X = c - 1) - c\Pr(X \le c - 1))} \\
    &=
    (1 - \beta) \cdot \bracks*{\Pr(X \le c - 1) (\beta x + (1 - \beta) c - \beta x + \beta c) + \beta x \Pr(X = c - 1)} \\
    &=
    (1 - \beta) \cdot \bracks*{\Pr(X \le c - 1) c + \beta x \Pr(X = c - 1)} \\
    &=
    (1 - \beta) \cdot c \cdot \bracks*{\Pr(X \le c - 1) + \beta \Pr(X = c)}.
\end{align*}
Therefore,
\[
h'(x)
    =
    \frac{(1 - \beta) c \bracks*{\Pr(X \le c - 1) + \beta \Pr(X = c)}}{(\beta x + (1-\beta)c)^2} \ge 0,
\]
which implies that $\alpha_{\varphi}(x ; c, \beta)$ is increasing for $x \ge c$.
This finishes the case analysis and completes the proof.
\end{proof}

\subsection{Proof of \Cref{thm:isoelastic_ratio}}

Let $\varphi(x) = x^{1-\gamma}$ for $\gamma \in (0, 1)$.
\citet[Proposition B.3]{barman2022tight} show that
\begin{align}
\label{eqn:isoelastic_series}
    \alpha_{\varphi}
    &=
    \frac{1}{e} \sum_{k=1}^\infty \frac{k^{1-\gamma}}{k!}.
\end{align}
We build on their result to give a closed-form formula for the Poisson concavity ratio.

\Isoelastic*

\begin{proof}
We manipulate the integral above to match the infinite series in \eqref{eqn:isoelastic_series}.
Start with the substitution $u = - \log x$.
Then $x = e^{-u}$ and $dx = - e^{-u} du$.
It follows that
\begin{align*}
    \int_{0}^{1} \frac{e^x}{(-\log x)^{1-\gamma}} dx
    &=
    \int_{\infty}^{0} \frac{e^{e^{-u}}}{u^{1-\gamma}} \parens*{-e^{-u}} du \\
    &=
    \int_{0}^{\infty} \frac{e^{e^{-u}} e^{-u} }{u^{1-\gamma}} du.
\end{align*}
Using the Taylor series expansion for $e^x$, we have
\begin{align*}
    e^{e^{-u}}
    &=
    \sum_{k=0}^\infty \frac{\parens*{e^{-u}}^k}{k!}
    =
    \sum_{k=0}^\infty \frac{e^{-ku}}{k!}.
\end{align*}
Substituting this series back into the integral,
\begin{align*}
    \int_{0}^{\infty} \frac{1}{u^{1-\gamma}} \parens*{\sum_{k=0}^\infty \frac{e^{-ku}}{k!}} e^{-u} du
    &=
    \int_{0}^{\infty} \sum_{k=0}^\infty \frac{e^{-(k + 1)u}}{k! u^{1-\gamma}} du.
\end{align*}
Swapping the sum and integral and letting $n = k + 1$, we have
\begin{align*}
    \sum_{k=0}^\infty \frac{1}{k!} \int_{0}^{\infty} u^{-(1-\gamma)} e^{-(k+1)u} du
    &=
    \sum_{n=1}^\infty \frac{1}{(n-1)!} \int_{0}^{\infty} u^{-(1-\gamma)} e^{-nu} du.
\end{align*}

To evaluate the integral, let $t = nu$, which implies $u = t/n$ and $du = dt / n$.
Therefore,
\begin{align*}
    \int_{0}^\infty \parens*{\frac{t}{n}}^{-(1-\gamma)} e^{-t} \frac{dt}{n}
    &=
    \frac{1}{n^{1 - (1-\gamma)}} \int_{0}^\infty t^{-(1-\gamma)} e^{-t} dt \\
    &=
    \frac{1}{n^{\gamma}} \cdot \Gamma(\gamma).
\end{align*}
Putting everything together,
\begin{align*}
    \int_{0}^{1} \frac{e^x}{(-\log x)^{1-\gamma}} dx
    &=
    \sum_{n=1}^\infty \frac{1}{(n-1)!} \cdot \frac{\Gamma(\gamma)}{n^{\gamma}} \\
    &=
    \Gamma(\gamma) \sum_{n=1}^\infty \frac{n^{1-\gamma}}{n!}.
\end{align*}
It follows that the full integral expression is
\begin{align*}
    \frac{1}{e \Gamma(\gamma)} \int_{0}^{1} \frac{e^x}{(-\log x)^{1-\gamma}} dx
    &=
    \frac{1}{e \Gamma(\gamma)} \cdot \Gamma(\gamma) \sum_{n=1}^\infty \frac{n^{1-\gamma}}{n!} \\
    &=
    \frac{1}{e} \sum_{n=1}^\infty \frac{n^{1-\gamma}}{n!},
\end{align*}
which completes the proof.
\end{proof}

%% file: app_greedy.tex
\section{Greedy algorithm upper bound}
\label{app:hard_greedy_instances}

\HardInstances*

\begin{proof}
Consider the following family of $c$-multi-coverage instances, i.e., with reward function $\varphi(x) = \min(x, c)$.

\paragraph{Setup.}
We set integer parameters $\ell = c -1$ and $q = c+1$. We set the cardinality constraint $k = q \cdot c$.  
This setting ensures that $k = q \cdot c = (\ell+2)(\ell+1)$ is greater than $(\ell+1)^2$. This is an inequality we will need later in the construction of the instance. 

We first explain the structure of the optimum solution.
Parameter $\ell$ represents the marginal gain of every left node in the optimum solution (they are disjoint and not overlapping).

\paragraph{Construction.}
 We start the construction of the instance based on \Cref{fig:MatrixGreedyInstance}.
 It is a matrix with~$\ell q$ rows and $c q$ columns consisting of $q$ rectangular submatrices stacked in a diagonal form. 
 
 Each right node in $R$ is represented by a row of this matrix, so $|R| = \ell q$. 
 Each column of this matrix represents a left node in the optimum solution. We note that the number of columns, $c q$, is precisely equal to $k$. So each left node in the optimum solution (each column) covers exactly $\ell$ right nodes. The reward function ensures that each of these $k$ left nodes provide a disjoint value of $\ell$ without any overlap. So the optimum solution covers each right node up to the saturation point of $c$ providing a total value of $\OPT = \ell q c$. 

\begin{figure}[h]
    \centering
    \begin{tikzpicture}
        
        \draw[thick, fill=blue!10] (0.5, 4.5) rectangle (1.5, 5.5) node[midway] {$M_1$};
        \node at (0.3, 5) {$\ell$}; 
        \node at (1.05, 5.7) {$c$}; 
        
        \draw[thick, fill=blue!10] (1.5, 3.5) rectangle (2.5, 4.5) node[midway] {$M_2$};
        \node at (1.3, 4) {$\ell$};
        \node at (2.05, 4.7) {$c$};
        
        \node at (2.5, 3) {$\ddots$};
        
        \draw[thick, fill=blue!10] (3.5, 1) rectangle (4.5, 2) node[midway] {$M_q$};
        \node at (3.3, 1.5) {$\ell$};
        \node at (4.05, 2.2) {$c$};
        
        \draw[decorate,decoration={brace,amplitude=10pt}] (4.6, 5.5) -- (4.6, 1) node [midway,xshift=1cm] {$q$ boxes};
        
    \end{tikzpicture}
    \caption{Construction of the optimal solution (submatrices $M_1, M_2, \dots M_q$).}
    \label{fig:MatrixGreedyInstance}
\end{figure}
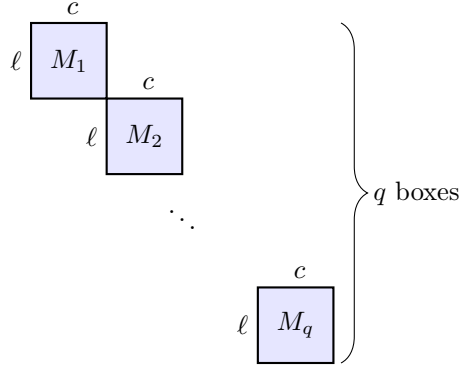

\paragraph{Greedy sets.}
We also add the following left nodes (sets) and show that the greedy algorithm will choose them instead of the optimum left nodes. 

From each submatrix, we pick an arbitrary right node, e.g., $x_1 \in M_1, x_2 \in M_2, \dots, x_q \in M_q$.
The greedy left nodes, $y_i \in L$, are defined as follows:
\begin{align*}
    y_1 &= \{ x_1, \dots, x_{\ell+1} \} \\
    y_2 &= \{ x_{\ell+2}, \dots, x_{2\ell+2} \} \\
        &\hspace{1.35cm}\vdots
\end{align*}

Each choice $y_i$ covers exactly $\ell+1$ new nodes. We continue this process until each node is covered exactly $c$ times. Each blue arc in \Cref{fig:CyclicCoverage} represents a left node we add and we expect the greedy algorithm to select.
We will explain shortly that we may need to increase the length of some of these arcs to $\ell+2$.
To ensure each arc does not go around the circle more than a full cycle, we need $\ell+2 \leq q$ as another parameter setting assumption.
Since we want each of the $\ell$ right nodes to be covered $c$ times, we approximately need $q c / (\ell+1)$ arcs.
We will add exactly $\lfloor q c / (\ell+1) \rfloor$ left nodes and make each of them cover $\ell+1$ or $\ell+2$ right nodes. If the remainder of dividing $q c$ by $\ell+1$ is $r$, we will need $r$ arcs of length $\ell+2$ and the rest should be of length $\ell+1$. 
We note that this remainder could be as large as $\ell$. To be able to pick $\ell$ arcs out of the total number of $\lfloor q c / (\ell+1) \rfloor$ arcs, we simply need the inequality $\ell \leq \lfloor q c / (\ell+1) \rfloor$ to hold. It suffices to have $\ell+1 \leq q c / (\ell+1)$ or equivalently $(\ell+1)^2 \leq qc = k$ which is one of requirements in the parameter setting. 
 
\begin{figure}[h]
    \centering
    \begin{tikzpicture}
        \draw[thick] (0,0) circle (2cm);
        
        \node[circle, fill=black, inner sep=1.5pt, label=above:$x_1$] at (75:2cm) {};
        \node[circle, fill=black, inner sep=1.5pt, label=right:$x_{\ell+1}$] at (45:2cm) {};
        \node[circle, fill=black, inner sep=1pt] at (65:2cm) {};
        \node[circle, fill=black, inner sep=1pt] at (55:2cm) {};
        
        \node[circle, fill=black, inner sep=1.5pt, label=above:$x_q$] at (90:2cm) {};

        \draw[thick, blue, ->] (80:2.6) arc (80:40:2.6);
        \node[blue] at (60:2.8) {$y_1$};
        
        \draw[thick, blue, ->] (27:3.1) arc (27:-5:3.1);
        \node[blue] at (15:3.3) {$y_2$};        
        \node[circle, fill=black, inner sep=1.5pt, label=right:$x_{\ell+2}$] at (30:2cm) {};
        \node[circle, fill=black, inner sep=1.5pt, label=right:$x_{2\ell+2}$] at (0:2cm) {};
        \node[circle, fill=black, inner sep=1pt] at (20:2cm) {};
        \node[circle, fill=black, inner sep=1pt] at (10:2cm) {};
        
        \node at (-45:2.3) {$\dots$};
        
    \end{tikzpicture}
    \caption{Cyclic coverage of nodes by greedy sets $y_i$.}
    \label{fig:CyclicCoverage}
\end{figure}
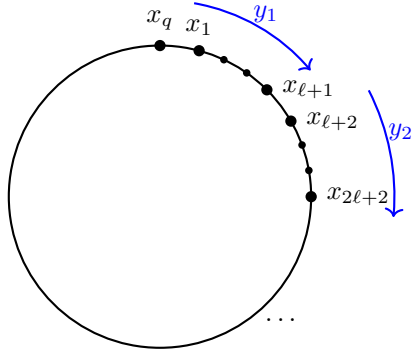

Since the marginal values of these arcs are either $\ell + 1$ or $\ell +2$, greedy will choose them over the optimum sets/columns that have marginal value of $\ell$. After all these 
$\lfloor \frac{q c}{\ell+1} \rfloor$ sets are chosen by greedy, the right nodes $x_1, x_2, \dots, x_q$ will be completely saturated with zero marginal value left on them. This means at this point, the marginal value of any optimum set is decreased by one unit, $\ell - 1$ instead of $\ell$. We announce this as the first phase of greedy sets construction, and move on to the next phase. We repeat a similar construction but this time, we only need the arcs to have length $\ell$ or $\ell + 1$. In the second phase, we will add $\lfloor \frac{q c}{\ell} \rfloor$ greedy sets, and this process continues. 
At the $i$-th phase, we add $\approx \frac{q c}{\ell+2-i}$ left nodes and greedy will choose them. 

Recall that $q c$ is equal to $k$, the cardinality constraint. Let's analyze what happens when greedy runs out of cardinality constraint budget. 
We need to find the smallest value of $i^*$ (or a good upper bound on it) for which
\[
    \sum_{i=1}^{i^*} \floor*{\frac{k}{\ell+2-i}} \geq k.
\]

For $1 \leq i' \leq \ell + 1$, we have
\begin{align*}
    \sum_{i=1}^{i'} \floor*{\frac{k}{\ell+2-i}} &\geq  \parens*{\sum_{i=1}^{i'} \frac{k}{\ell+2-i} } - i' \\
    &=  \parens*{\sum_{i=1}^{i'-1} \frac{k}{\ell+2-i} } + \frac{k}{\ell+2-i'} - i' \\
    &\geq  \sum_{i=1}^{i'-1} \frac{k}{\ell+2-i} \\
    &= k \sum_{j=\ell+3-i'}^{\ell+1} \frac{1}{j} \\
    &\geq k \int_{x=\ell+3-i'}^{\ell+2} \frac{1}{x} \, dx  \\
    &= k \log \parens*{\frac{\ell+2}{\ell+3-i'}}, \\
\end{align*}
where the second inequality, i.e., $\frac{k}{\ell+2-i'} - i' \geq 0$, holds because of the parameter setting assumption that $k \geq (\ell+1)^2$.
For the above lower bound to be at least $k$, it suffices to have $\log \frac{\ell+2}{\ell+3-i'} \geq 1$ or equivalently $\frac{\ell+2}{e} \geq \ell+3-i'$, where $e \approx 2.71$ is the Euler number.
Therefore, it suffices to set
\[
    i^* = \ell+3 - \floor*{\frac{\ell+2}{e}}.
\]

The greedy solution covers each right node at most $i^* = \ell+3 - \floor*{\frac{\ell+2}{e}}$ compared to the optimum solution in which right nodes are covered $\ell$ times. The approximation ratio of algorithm greedy on this instance is at most
\[ 
    \frac{\ell+3 - \floor*{\frac{\ell+2}{e}}}{\ell} \leq \frac{\ell+3 - \frac{\ell+2}{e} + 1}{\ell} \leq 1 - \frac{1}{e} + \frac{3.3}{\ell}.
\]
We recall the parameter setting: $c = \ell+1$, $q= \ell+2$, and $k = qc \geq (\ell+1)^2$.
The loss term in the approximation bound is $\frac{3.3}{\ell} = \frac{3.3}{c-1}$.
Since $c \ge 1 + \ceil*{\frac{3.3}{\varepsilon}}$, the approximation guarantee of the greedy algorithm is at most $1 - \sfrac{1}{e} + \varepsilon$.
We also note that the approximation guarantee of the greedy algorithm in this construction is 
 $1 - \sfrac{1}{e} + O(\frac{1}{\sqrt{k}})$ since $\ell \geq \sqrt{k} - 2$,
 which completes the proof.
\end{proof}

%% file: app_experiments.tex
\section{Additional experiments}
\label{app:experiments}

All of the experiments are implemented in C\texttt{++} and run on an AMD EPYC 7B13 processor (2.45 GHz, 256MB L3 cache) with 240GB RAM.
The LP solver versions we use are GLOP v9.15~\citep{ortools}, HiGHS v.1.10.0~\citep{huangfu2018parallelizing}, and SCIP v7.0.1~\citep{hojny2025scip}.

We conducted our experiments on two graphs in the Stanford Network Analysis Platform (SNAP)~\citep{snapnets}: \facebook and \dblp.
The original \facebook graph has 4,039 nodes and \mytilde 88k edges,
and the original \dblp graph has of 317,080 nodes and \mytilde 1M edges.
For each dataset, we construct a symmetric bipartite graph $G' = (L, R, E')$.
Specifically, for each node $u$ in the original graph, we create two copies: $u \in L$ and $u' \in R$, connected by an edge $\{u, u'\}$. Further, each original undirected edge $\{u, v\}$ is mapped to two edges, $\{u, v'\}$ and $\{v, u'\}$, in the bipartite coverage graph.
Consequently, the resulting bipartite graph for \facebook has 4,039 nodes per side with 180,507 total edges,
and the \dblp bipartite graph has 317,080 nodes per side and 2,416,812 total edges.
We study $c$-multi-coverage~\citep{barman2022tight}, where $\varphi(x) = \min(x, c)$,
for different values of $c$ and cardinality constraints $k$.

\paragraph{Facebook.}

\begin{table}
  \caption{Running time in seconds for different relax-and-round algorithms on the \facebook graph.}
  \label{table:facebook_running_times}
  \centering
  \begin{tabular}{rrrrrrrrrrrrrrrrrrr}
    \toprule
    & & \multicolumn{3}{c}{\textbf{LP Solver + PipageRounding}} & \multicolumn{3}{c}{\textbf{AcceleratedRelaxAndRound}} \\
    \cmidrule(r){3-5} \cmidrule(l){6-8}
    & $k$ & GLOP & HiGHS & SCIP & $\eta = 10$ & $\eta = 1$ & $\eta = 0.1$ \\

    \midrule
    \multirow{5}{*}{\rotatebox[origin=c]{90}{$c=1$}}
    & 20 & 0.70 \scriptsize{$\pm$ 0.01} & 0.42 \scriptsize{$\pm$ 0.03} & 3.03 \scriptsize{$\pm$ 0.02} & 0.05 \scriptsize{$\pm$ 0.00} & 0.34 \scriptsize{$\pm$ 0.00} & \textbf{0.04} \scriptsize{$\pm$ 0.00} \\
    & 40 & 0.70 \scriptsize{$\pm$ 0.00} & 0.40 \scriptsize{$\pm$ 0.01} & 3.04 \scriptsize{$\pm$ 0.02} & \textbf{0.06} \scriptsize{$\pm$ 0.00} & 0.29 \scriptsize{$\pm$ 0.00} & 0.23 \scriptsize{$\pm$ 0.00} \\
    & 60 & 0.91 \scriptsize{$\pm$ 0.02} & 0.40 \scriptsize{$\pm$ 0.00} & 3.03 \scriptsize{$\pm$ 0.02} & \textbf{0.05} \scriptsize{$\pm$ 0.00} & 0.26 \scriptsize{$\pm$ 0.00} & 0.25 \scriptsize{$\pm$ 0.00} \\
    & 80 & 0.96 \scriptsize{$\pm$ 0.02} & 0.41 \scriptsize{$\pm$ 0.01} & 3.07 \scriptsize{$\pm$ 0.03} & \textbf{0.05} \scriptsize{$\pm$ 0.00} & 0.24 \scriptsize{$\pm$ 0.00} & 0.25 \scriptsize{$\pm$ 0.00} \\
    & 100 & 0.68 \scriptsize{$\pm$ 0.01} & 0.40 \scriptsize{$\pm$ 0.01} & 6.61 \scriptsize{$\pm$ 1.17} & \textbf{0.05} \scriptsize{$\pm$ 0.00} & 0.22 \scriptsize{$\pm$ 0.00} & 0.25 \scriptsize{$\pm$ 0.00} \\

    \midrule
    \multirow{5}{*}{\rotatebox[origin=c]{90}{$c=2$}} &
    20 & 0.84 \scriptsize{$\pm$ 0.01} & 1.12 \scriptsize{$\pm$ 0.01} & 7.64 \scriptsize{$\pm$ 0.14} & \textbf{0.04} \scriptsize{$\pm$ 0.00} & 0.05 \scriptsize{$\pm$ 0.00} & 0.98 \scriptsize{$\pm$ 0.01} \\
    & 40 & 4.00 \scriptsize{$\pm$ 0.01} & 4.34 \scriptsize{$\pm$ 0.02} & 10.60 \scriptsize{$\pm$ 0.47} & \textbf{0.04} \scriptsize{$\pm$ 0.00} & 0.08 \scriptsize{$\pm$ 0.00} & 1.46 \scriptsize{$\pm$ 0.00} \\
    & 60 & 6.76 \scriptsize{$\pm$ 0.06} & 7.02 \scriptsize{$\pm$ 0.04} & 13.66 \scriptsize{$\pm$ 1.19} & \textbf{0.04} \scriptsize{$\pm$ 0.00} & 0.07 \scriptsize{$\pm$ 0.00} & 2.08 \scriptsize{$\pm$ 0.01} \\
    & 80 & 16.58 \scriptsize{$\pm$ 0.11} & 16.56 \scriptsize{$\pm$ 0.03} & 22.47 \scriptsize{$\pm$ 2.08} & \textbf{0.05} \scriptsize{$\pm$ 0.00} & 0.09 \scriptsize{$\pm$ 0.00} & 3.81 \scriptsize{$\pm$ 0.02} \\
    & 100 & 20.28 \scriptsize{$\pm$ 0.14} & 20.74 \scriptsize{$\pm$ 1.21} & 24.97 \scriptsize{$\pm$ 0.71} & \textbf{0.05} \scriptsize{$\pm$ 0.00} & 0.14 \scriptsize{$\pm$ 0.00} & 3.95 \scriptsize{$\pm$ 0.02} \\

    \midrule
    \multirow{5}{*}{\rotatebox[origin=c]{90}{$c=4$}}
    & 20 & 0.35 \scriptsize{$\pm$ 0.01} & 0.69 \scriptsize{$\pm$ 0.01} & 6.05 \scriptsize{$\pm$ 0.04} & \textbf{0.04} \scriptsize{$\pm$ 0.00} & \textbf{0.04} \scriptsize{$\pm$ 0.00} & 0.19 \scriptsize{$\pm$ 0.00} \\
    & 40 & 0.79 \scriptsize{$\pm$ 0.00} & 1.37 \scriptsize{$\pm$ 0.02} & 6.80 \scriptsize{$\pm$ 0.09} & \textbf{0.04} \scriptsize{$\pm$ 0.00} & 0.05 \scriptsize{$\pm$ 0.00} & 2.61 \scriptsize{$\pm$ 0.00} \\
    & 60 & 1.23 \scriptsize{$\pm$ 0.00} & 1.72 \scriptsize{$\pm$ 0.01} & 7.05 \scriptsize{$\pm$ 0.11} & \textbf{0.04} \scriptsize{$\pm$ 0.00} & 0.05 \scriptsize{$\pm$ 0.00} & 1.94 \scriptsize{$\pm$ 0.01} \\
    & 80 & 2.41 \scriptsize{$\pm$ 0.02} & 2.45 \scriptsize{$\pm$ 0.01} & 8.39 \scriptsize{$\pm$ 0.20} & \textbf{0.04} \scriptsize{$\pm$ 0.00} & 0.05 \scriptsize{$\pm$ 0.00} & 3.65 \scriptsize{$\pm$ 0.02} \\
    & 100 & 4.79 \scriptsize{$\pm$ 0.04} & 5.04 \scriptsize{$\pm$ 0.02} & 10.04 \scriptsize{$\pm$ 0.10} & \textbf{0.04} \scriptsize{$\pm$ 0.00} & 0.07 \scriptsize{$\pm$ 0.00} & 3.29 \scriptsize{$\pm$ 0.01} \\
    \bottomrule
  \end{tabular}
\end{table}

In \Cref{table:facebook_objective_values} and \Cref{table:dblp_objective_values},
we report the maximum rounded objective value over trials for each algorithm, and also compare against the greedy objective value.
In each row, we boldface the maximum value in the row if there is a strictly smaller value in that row (ignoring TLEs).

\begin{table}
  \caption{Rounded objective values for different relax-and-round algorithms on the \facebook graph.}
  \label{table:facebook_objective_values}
  \centering
  \begin{tabular}{rrrrrrrrrrrrrrrrrrr}
    \toprule
    & & & \multicolumn{3}{c}{\textbf{LP + PipageRounding}} & \multicolumn{3}{c}{\textbf{AcceleratedRelaxAndRound}} \\
    \cmidrule(r){4-6} \cmidrule(l){7-9}
    & $k$ & Greedy & GLOP & HiGHS & SCIP & $\eta = 10$ & $\eta = 1$ & $\eta = 0.1$ \\
    
    \midrule
    \multirow{5}{*}{\rotatebox[origin=c]{90}{$c=1$}}
    &  20 & 4,039 & 4,039 & 4,039 & 4,039 & 4,039 & 4,039 & 4,039 \\
    &  40 & 4,039 & 4,039 & 4,039 & 4,039 & 4,039 & 4,039 & 4,039 \\
    &  60 & 4,039 & 4,039 & 4,039 & 4,039 & 4,039 & 4,039 & 4,039 \\
    &  80 & 4,039 & 4,039 & 4,039 & 4,039 & 4,039 & 4,039 & 4,039 \\
    & 100 & 4,039 & 4,039 & 4,039 & 4,039 & 4,039 & 4,039 & 4,039 \\

    \midrule
    \multirow{5}{*}{\rotatebox[origin=c]{90}{$c=2$}}
    &  20 & 5,607 & \textbf{5,653} & \textbf{5,653} & \textbf{5,653} & 4,643 & 5,607 & 5,649 \\
    &  40 & 6,488 & 6,486 & \textbf{6,498} & 6,486 & 6,435 & 6,488 & 6,452 \\
    &  60 & \textbf{6,931} & 6,928 & 6,928 & 6,928 & 6,887 & \textbf{6,931} & 6,909 \\
    &  80 & \textbf{7,196} & 7,162 & 7,183 & 7,170 & 7,144 & \textbf{7,196} & 7,131 \\
    & 100 & \textbf{7,377} & 7,360 & 7,369 & 7,344 & 7,309 & 7,370 & 7,108 \\

    \midrule
    \multirow{5}{*}{\rotatebox[origin=c]{90}{$c=4$}}
    &  20 &  6,551 &  6,551 &  6,551 &  6,551 &  6,551 &   6,551 &  6,551 \\
    &  40 &  8,428 &  8,457 &  8,457 &  8,457 &  8,428 &   8,403 &  \textbf{8,460} \\
    &  60 &  9,655 &  \textbf{9,701} &  \textbf{9,701} &  \textbf{9,701} &  9,638 &   9,660 &  9,693 \\
    &  80 & 10,507 & 10,558 & \textbf{10,559} & 10,556 & 10,456 &  10,507 & 10,497 \\
    & 100 & 11,151 & 11,188 & \textbf{11,192} & 11,187 & 11,137 &  11,151 & 11,131 \\
    \bottomrule
  \end{tabular}
\end{table}

\paragraph{DBLP.}

The two tables for the \dblp coverage graph are analogous to those in the \facebook section.

\begin{table}[H]
  \caption{Running time in seconds for different relax-and-round algorithms on the \dblp graph.}
  \label{table:dblp_running_times}
  \centering
  \begin{tabular}{rrrrrrrrrrrrrrrrrrr}
    \toprule
    & & \multicolumn{3}{c}{\textbf{LP Solver + PipageRounding}} & \multicolumn{3}{c}{\textbf{AcceleratedRelaxAndRound}} \\
    \cmidrule(r){3-5} \cmidrule(l){6-8}
    & $k$ & GLOP & HiGHS & SCIP & $\eta = 10$ & $\eta = 1$ & $\eta = 0.1$ \\
    
    \midrule
    \multirow{5}{*}{\rotatebox[origin=c]{90}{$c=1$}}
    &  20 & 141.45 \scriptsize{$\pm$ 0.31} & TLE & TLE & \textbf{1.60} \scriptsize{$\pm$ 0.03} & 2.28 \scriptsize{$\pm$ 0.05} & 3.81 \scriptsize{$\pm$ 0.15} \\
    &  40 & 698.79 \scriptsize{$\pm$ 0.51} & TLE & TLE & \textbf{1.74} \scriptsize{$\pm$ 0.05} & 2.30 \scriptsize{$\pm$ 0.06} & 4.87 \scriptsize{$\pm$ 0.14} \\
    &  60 & 736.95 \scriptsize{$\pm$ 0.71} & TLE & TLE & \textbf{1.72} \scriptsize{$\pm$ 0.03} & 2.51 \scriptsize{$\pm$ 0.08} & 5.16 \scriptsize{$\pm$ 0.20} \\
    &  80 & 739.94 \scriptsize{$\pm$ 0.55} & TLE & TLE & \textbf{1.78} \scriptsize{$\pm$ 0.03} & 3.11 \scriptsize{$\pm$ 0.07} & 6.53 \scriptsize{$\pm$ 0.18} \\
    & 100 & 624.37 \scriptsize{$\pm$ 0.55} & TLE & TLE & \textbf{1.76} \scriptsize{$\pm$ 0.04} & 4.95 \scriptsize{$\pm$ 0.13} & 7.02 \scriptsize{$\pm$ 0.17} \\

    \midrule
    \multirow{5}{*}{\rotatebox[origin=c]{90}{$c=2$}}
    & 20 & 105.79 \scriptsize{$\pm$ 0.38} & 50.45 \scriptsize{$\pm$ 0.63} & TLE & 1.42 \scriptsize{$\pm$ 0.05} & \textbf{1.38} \scriptsize{$\pm$ 0.05} & 2.09 \scriptsize{$\pm$ 0.09} \\
    & 40 & 116.98 \scriptsize{$\pm$ 0.28} & 56.36 \scriptsize{$\pm$ 1.39} & TLE & 1.58 \scriptsize{$\pm$ 0.02} & 1.79 \scriptsize{$\pm$ 0.04} & \textbf{1.48} \scriptsize{$\pm$ 0.04} \\
    & 60 & 665.57 \scriptsize{$\pm$ 0.39} & 78.01 \scriptsize{$\pm$ 1.68} & TLE & 1.53 \scriptsize{$\pm$ 0.04} & 1.80 \scriptsize{$\pm$ 0.05} & \textbf{1.51} \scriptsize{$\pm$ 0.02} \\
    & 80 & 733.57 \scriptsize{$\pm$ 0.98} & 75.13 \scriptsize{$\pm$ 1.55} & TLE & 1.53 \scriptsize{$\pm$ 0.00} & 2.24 \scriptsize{$\pm$ 0.07} & \textbf{1.48} \scriptsize{$\pm$ 0.01} \\
    & 100 & 732.11 \scriptsize{$\pm$ 1.40} & 86.70 \scriptsize{$\pm$ 1.53} & TLE & 1.56 \scriptsize{$\pm$ 0.02} & 2.28 \scriptsize{$\pm$ 0.02} & \textbf{1.47} \scriptsize{$\pm$ 0.01} \\

    \midrule
    \multirow{5}{*}{\rotatebox[origin=c]{90}{$c=4$}}
    &  20 & 15.91 \scriptsize{$\pm$ 0.04} & 30.62 \scriptsize{$\pm$ 0.54} & TLE & 0.81 \scriptsize{$\pm$ 0.01} & \textbf{0.52} \scriptsize{$\pm$ 0.00} & \textbf{0.52} \scriptsize{$\pm$ 0.01} \\
    &  40 & 17.19 \scriptsize{$\pm$ 0.20} & 31.33 \scriptsize{$\pm$ 0.51} & TLE & 1.27 \scriptsize{$\pm$ 0.02} & \textbf{1.26} \scriptsize{$\pm$ 0.01} & 1.56 \scriptsize{$\pm$ 0.01} \\
    &  60 & 617.37 \scriptsize{$\pm$ 0.32} & 31.61 \scriptsize{$\pm$ 0.48} & TLE & 1.68 \scriptsize{$\pm$ 0.03} & 1.61 \scriptsize{$\pm$ 0.01} & \textbf{1.60} \scriptsize{$\pm$ 0.03} \\
    &  80 & 618.39 \scriptsize{$\pm$ 0.74} & 31.76 \scriptsize{$\pm$ 0.59} & TLE & \textbf{1.60} \scriptsize{$\pm$ 0.01} & 2.21 \scriptsize{$\pm$ 0.01} & 1.74 \scriptsize{$\pm$ 0.01} \\
    & 100 & 613.05 \scriptsize{$\pm$ 0.61} & 32.06 \scriptsize{$\pm$ 0.57} & TLE & \textbf{1.56} \scriptsize{$\pm$ 0.00} & 1.83 \scriptsize{$\pm$ 0.01} & 1.67 \scriptsize{$\pm$ 0.02} \\
    \bottomrule
  \end{tabular}
\end{table}

\begin{table}[H]
  \caption{Rounded objective values for different relax-and-round algorithms on the \dblp graph.}
  \label{table:dblp_objective_values}
  \centering
  \begin{tabular}{rrrrrrrrrrrrrrrrrrr}
    \toprule
    & & & \multicolumn{3}{c}{\textbf{LP Solver + PipageRounding}} & \multicolumn{3}{c}{\textbf{AcceleratedRelaxAndRound}} \\
    \cmidrule(r){4-6} \cmidrule(l){7-9}
    & $k$ & Greedy & GLOP & HiGHS & SCIP & $\eta = 10$ & $\eta = 1$ & $\eta = 0.1$ \\
    
    \midrule
    \multirow{5}{*}{\rotatebox[origin=c]{90}{$c=1$}}
    &  20 &  4,356 &  4,356 & TLE & TLE &  4,356 &  4,356 &  4,356 \\
    &  40 &  \textbf{7,467} &  6,626 & TLE & TLE &  \textbf{7,467} &  \textbf{7,467} &  \textbf{7,467} \\
    &  60 & \textbf{10,098} &  8,624 & TLE & TLE & \textbf{10,098} & \textbf{10,098} & \textbf{10,098} \\
    &  80 & \textbf{12,430} & 10,488 & TLE & TLE & \textbf{12,430} & \textbf{12,430} & \textbf{12,430} \\
    & 100 & \textbf{14,567} &  7,195 & TLE & TLE & \textbf{14,567} & \textbf{14,567} & \textbf{14,567} \\
    
    \midrule
    \multirow{5}{*}{\rotatebox[origin=c]{90}{$c=2$}}
    &  20 &  4,625 &  4,625 &  4,625 & TLE &  4,625 &  4,625 &  4,625 \\
    &  40 &  8,074 &  8,074 &  8,074 & TLE &  8,074 &  8,074 &  8,074 \\
    &  60 & \textbf{11,056} &  9,737 & \textbf{11,056} & TLE & \textbf{11,056} & \textbf{11,056} & \textbf{11,056} \\
    &  80 & \textbf{13,736} & 12,407 & \textbf{13,736} & TLE & \textbf{13,736} & \textbf{13,736} & \textbf{13,736} \\
    & 100 & \textbf{16,214} & 14,288 & \textbf{16,214} & TLE & \textbf{16,214} & \textbf{16,214} & \textbf{16,214} \\

    \midrule
    \multirow{5}{*}{\rotatebox[origin=c]{90}{$c=4$}}
    &  20 &  4,723 &  4,723 &  4,723 & TLE & 4,723 & 4,723 & 4,723 \\
    &  40 &  8,366 &  8,366 &  8,366 & TLE & 8,366 & 8,366 & 8,366 \\
    &  60 & \textbf{11,525} & 10,822 & \textbf{11,525} & TLE & \textbf{11,525} & \textbf{11,525} & \textbf{11,525} \\
    &  80 & \textbf{14,407} & 12,581 & \textbf{14,407} & TLE & \textbf{14,407} & \textbf{14,407} & 12,993 \\
    & 100 & \textbf{17,088} & 14,921 & \textbf{17,088} & TLE & \textbf{17,088} & \textbf{17,088} & \textbf{17,088} \\
\bottomrule
  \end{tabular}
\end{table}

\paragraph{Larger values of $k$.}
If $c = 1$ and $k$ spans a larger range, \Cref{alg:accelerated_relax_and_round} outperforms greedy in terms of objective value (\Cref{table:dblp_objective_values_large_k_c1}).
For these larger instances, SCIP and HiGHS \emph{fail to find a feasible solution} with a one-hour time limit.
GLOP finds a feasible solution for $k \in \{2000, 5000, 10000\}$, but the respective objective values after pipage rounding
are (33,008, 40,965, 64,556), which are significantly worse than greedy and \Cref{alg:accelerated_relax_and_round}.
In the following experiment, we amplified the quality of \SwapRound by rounding 200 times since \SwapRound is orders of magnitude cheaper than solving for $\widetilde{x}^*$.
This is not applicable to pipage rounding.
We present the running times of \Cref{alg:accelerated_relax_and_round} in \Cref{table:dblp_running_times_large_k_c1}.  

\begin{table}[H]
  \caption{Comparison of objective values achieved by \Cref{alg:accelerated_relax_and_round} before and after rounding on \dblp
  for $c = 1$ and larger values of $k$.}
  \label{table:dblp_objective_values_large_k_c1}
  \centering
  \begin{tabular}{rrrrrrrrrrrrrrrrrrr}
    \toprule
    & & &  \multicolumn{3}{c}{$C(\widetilde{x}^*)$}  & \multicolumn{3}{c}  {$C(\SwapRound(\widetilde{x}^*))$} \\
    \cmidrule(l){4-6}\cmidrule(l){7-9}
    & $k$ & Greedy & $\eta = 10$ & $\eta = 1$ & $\eta = 0.1$ & $\eta = 10$ & $\eta = 1$ & $\eta = 0.1$ \\
    \midrule
    \multirow{5}{*}{\rotatebox[origin=c]{90}{}}
    &    200 &  23,362 &   23,372.37 &   23,375.49 &  23,362.00 &  \textbf{23,376} &  \textbf{23,376} &  23,362 \\
    &    500 &  42,627 &   42,691.05 &   42,688.53 &  42,675.02 &  42,665 &  42,681	&  \textbf{42,686} \\
    &  1,000 &  65,387 &   65,492.60 &   65,488.31 &  65,465.10	&  \textbf{65,484} &  65,447	&  65,459 \\
    &  2,000 &  97,179 &   97,320.78 &   97,313.96 &  97,288.22 &  \textbf{97,237} &  97,226 &  97,183 \\
    &  5,000 & 155,383 &  155,615.68 &  155,606.43 & 155,559.09 & \textbf{155,458} & 155,419 & 155,224\\
    & 10,000 & 209,679 &  210,098.03 &  210,086.98 & 210,016.43 & \textbf{209,748} & 209,685 & 209,476 \\
    \bottomrule
  \end{tabular}
\end{table}

\begin{table}[H]
  \caption{Running time in seconds for \Cref{alg:accelerated_relax_and_round} on \dblp for $c = 1$ and larger values of $k$.}
  \label{table:dblp_running_times_large_k_c1}
  \centering
  \begin{tabular}{rrrrrrrrrrrrrrrrrrr}
    \toprule
     & \multicolumn{3}{c}{\textbf{Solve}} & \multicolumn{3}{c}{\textbf{Rounding}} & \multicolumn{3}{c}{\textbf{Total}}\\
    \cmidrule(r){2-4} \cmidrule(l){5-7} \cmidrule(l){8-10}
    $k$ & $\eta = 10$ & $\eta = 1$ & $\eta = 0.1$ & $\eta = 10$ & $\eta = 1$ & $\eta = 0.1$ & $\eta = 10$ & $\eta = 1$ & $\eta = 0.1$ \\

    \midrule
    \multirow{5}{*}{\rotatebox[origin=c]{90}{}}
    200 & 82.70& 	185.16& 	23.66& 	1.92& 	2.10& 	0.75& 84.62	& 187.26 & 	24.41 \\
    500 & 107.56&	215.30&	302.49&	3.12&	3.18&	3.18& 110.68&	218.49 &	305.67 \\
    1,000&  117.79&	260.99& 	407.71 & 	5.37& 	5.27&	5.18& 123.17 & 	266.26	& 412.89\\
    2,000 &  127.10 & 	273.44& 	545.24&	8.23&	8.15&	8.05& 135.33	& 281.59	& 553.30\\
    5,000 & 136.42& 	337.86 & 	680.88 & 	17.91& 	17.62& 	16.99 & 154.32	& 355.47	& 697.87\\
    10,000 & 148.02 & 	398.67 & 	864.95& 	27.05& 	27.10 & 	28.99 & 175.08 & 	425.77& 	893.94\\
    \bottomrule
  \end{tabular}
\end{table}